\documentclass[review,hidelinks,onefignum,onetabnum]{siamart220329}

\usepackage{xcolor}
\usepackage{lipsum}
\usepackage{amsfonts}
\usepackage{graphicx}
\usepackage{epstopdf}
\usepackage{algorithm,algorithmic}
\ifpdf
  \DeclareGraphicsExtensions{.eps,.pdf,.png,.jpg}
\else
  \DeclareGraphicsExtensions{.eps}
\fi

\newsiamremark{remark}{Remark}
\newsiamremark{hypothesis}{Hypothesis}
\crefname{hypothesis}{Hypothesis}{Hypotheses}
\newsiamthm{claim}{Claim}

\headers{Signal-Comparison-Based Distributed Estimation}{J. M. Ke, X. D. Lu, Y. L. Zhao, and J. F. Zhang}

\title{Signal-Comparison-Based Distributed Estimation Under Decaying Average Data Rate Communications\thanks{Submitted to the editors DATE.
\funding{The work was supported by National Key R\&D Program of China under Grant 2018YFA0703800, National Natural Science Foundation of China under Grants T2293772 and 62025306, and CAS Project for Young Scientists in Basic Research under Grant YSBR-008.}}}

\author{Jieming Ke\thanks{Key Laboratory of Systems and Control, Institute of Systems Science, Academy of Mathematics and Systems Science, Chinese Academy of Sciences, Beijing 100190, People's Republic of China, and School of Mathematical Sciences, University of Chinese Academy of Sciences, Beijing 100049, People's Republic of China (kejieming@amss.ac.cn)} 
\and Xiaodong Lu\thanks{School of Automation and Electrical Engineering University of Science and Technology Beijing, Beijing 100083, People's Republic of China, and Key Laboratory of Knowledge Automation for Industrial Processes, Ministry of Education, Beijing 100083, People's Republic of China (lxd1025744209@163.com)}
\and Yanlong Zhao\thanks{Corresponding author. Key Laboratory of Systems and Control, Institute of Systems Science, Academy of Mathematics and Systems Science, Chinese Academy of Sciences, Beijing 100190, People's Republic of China, and School of Mathematical Sciences, University of Chinese Academy of Sciences, Beijing 100149, People's Republic of China (ylzhao@amss.ac.cn)}
\and Ji-Feng Zhang\thanks{School of Automation and Electrical Engineering, Zhongyuan University of Technology, Zhengzhou 450007,
Henan Province, People's Republic of China, and Key Laboratory of Systems and Control, Institute of Systems Science, Academy of Mathematics and	Systems Science, Chinese Academy of Sciences, Beijing 100190, People's Republic of China. (jif@iss.ac.cn)}
}
\usepackage{amsopn}
\DeclareMathOperator{\diag}{diag}

\ifpdf
\hypersetup{
  pdftitle={An Example Article},
  pdfauthor={D. Doe, P. T. Frank, and J. E. Smith}
}
\fi

\usepackage{lineno}
\usepackage{amssymb,amsmath}
\usepackage{mathrsfs}
\usepackage{enumitem}
\usepackage{subfloat}
\usepackage{algorithm,float}  
\usepackage{algorithmic}
\usepackage{upgreek}
\newsiamremark{assumption}{Assumption}
\crefname{section}{Section}{Sections}
\crefname{subsection}{Subsection}{Subsections}
\crefname{assumption}{Assumption}{Assumptions}
\crefname{appen}{Appendix}{Appendixes}
\def\tr{^\top}
\newcommand{\absl}[1]{\lvert #1 \rvert}
\newcommand{\Absl}[1]{\lVert #1 \rVert}
\newcommand{\abs}[1]{\left| #1\right|}
\newcommand{\Abs}[1]{\left\| #1\right\|}
\def\mE{\mathbb{E}}
\def\mP{\mathbb{P}}
\def\as{\text{a.s.}}

\def\col{\operatorname{col}}
\def\Bone{\mathbf{1}}

\def\Neighbour{\mathcal{N}}
\def\mL{\mathcal{L}}

\newcommand{\MODIFY}[1]{#1}

\usepackage{tikz}
\usetikzlibrary{shapes.geometric}
\usetikzlibrary{arrows,arrows.meta}

\tikzstyle{ellipse}=[draw, rectangle, minimum width=2.8em, rounded corners=6pt,line width=0.5pt]
\tikzstyle{pxsbx}=[trapezium, trapezium left angle=75, trapezium right angle=105, minimum width=3em, text centered, draw = black, fill=white,line width=0.5pt]
\tikzstyle{lingxing}=[draw,diamond,shape aspect=3,inner sep = 0.4pt,thick,font=\itshape,line width=0.5pt]
\usetikzlibrary{positioning, arrows.meta, calc}
\tikzset{
	arrow1/.style = {
		draw = black, thick, -{Latex[length = 4mm, width = 1.5mm]},
	}
}
\tikzset{
	nonterminal/.style = {
		rectangle,
		align = center,
		minimum size = 6mm,
		very thick,
		draw = red!50!black!50,
		top color = white,
		bottom color = red!50!black!20,
	}
}
\tikzset{
	terminal/.style = {
		rectangle,
		align = center,
		minimum size = 6mm,
		rounded corners = 3mm,
		very thick,
		draw = black!50,
		top color = white,
		bottom color = black!20,
	}
}

\begin{document}
	\linenumbers
	\nolinenumbers

\maketitle

\begin{abstract}
	The paper investigates the distributed estimation problem under low data rate communications. 
	Based on the signal-comparison (SC) consensus protocol under binary-valued communications, a new consensus+innovations type distributed estimation algorithm is proposed. Firstly, the high-dimensional estimates are compressed into binary-valued messages by using a periodic compressive strategy, dithering noises and a sign function. Next, based on the dithering noises and expanding triggering  thresholds, a new stochastic event-triggered mechanism is proposed to reduce the communication frequency. Then, a modified SC consensus protocol is applied to fuse the neighborhood information. Finally, a stochastic approximation estimation algorithm is used to process innovations. The proposed SC-based algorithm has the advantages of high effectiveness and low communication cost. For the effectiveness, the estimates of the SC-based algorithm converge to the true value in the almost sure and mean square sense, and a polynomial almost sure convergence rate is also obtained. For the communication cost, the local and global average data rates decay to zero at a polynomial rate. 
	The trade-off between the convergence rate and the communication cost is established through event-triggered coefficients. A better convergence rate can be achieved by decreasing event-triggered coefficients, while lower communication cost can be achieved by increasing event-triggered coefficients.  
	A simulation example is given to demonstrate the theoretical results. 
\end{abstract}



\section{Introduction}

Distributed estimation is of great practical significance in many practical fields, such as electric power grid \cite{kar2014distributed} and cognitive radio systems \cite{sayed2014diffusion}, and therefore has been being an attractive topic \cite{gan2022distributed,kar2013distributed,sahu2018CIRFE,wang2023differentially}.
In the distributed estimation problem, the subsystem of each sensor is not necessarily observable. Therefore, communications between sensors are required to fuse the observations of the distributed sensors, \MODIFY{which brings communication cost problems.} Firstly, 
\MODIFY{
due to the bandwidth limitations in the real digital networks, high data rate communications may cause network congestion.} 
Secondly, the transmission energy cost is positively correlated with the bit numbers of communication messages \cite{li2009distributed}. Therefore, it is important to propose a distributed estimation algorithm under low data rate communications.

There have been many works in quantization methods to reduce the communication cost for distributed algorithms \cite{aysal2008distributed,carli2009quantized,carli2008communication,kar2012distributed,zhu2018mitigating,zhu2015distributed}, many of which are based on infinity level quantizers. For example, Aysal et al. adopt infinite level probabilistic quantizers to construct a quantized consensus algorithm \cite{aysal2008distributed}. Furthermore, Carli et al. \cite{carli2009quantized,carli2008communication} propose an important technique based on infinite level logarithm quantizers to give quantized coordination algorithms and a quantized average consensus algorithm. Kar and Moura \cite{kar2012distributed} appear to be the first to consider distributed estimation under quantized communications. They improve the probabilistic quantizer-based consensus algorithm in \cite{aysal2008distributed} by using the stochastic approximation method. Based on the technique, the estimates of corresponding consensus+innovations distributed estimation algorithm converge to the true value. Besides, when there is only one observation for each sensor, Zhu et al.  \cite{zhu2018mitigating,zhu2015distributed} propose running average distributed estimation algorithms based on probabilistic quantizers. 

Due to the data rate limitations in real digital networks, distributed algorithms under finite data rate communications are developed. This is a challenging task because information contained in the interactive messages is limited. To solve the difficulty, Li et al. \cite{LiT2011distributed}, Liu et al. \cite{liu2011distributed}, and Meng et al. \cite{meng2017finite} design zooming-in methods for the consensus problems under finite data rate communications. The methods are effective to deal with the quantization error. When communication noises exist, Zhao et al. \cite{zhao2018consensus} and Wang et al. \cite{wang2019consensus} propose an empirical measurement-based consensus algorithm and a recursive projection consensus algorithm under binary-valued communications, respectively. 

Distributed estimation under finite data rate communications has also been extensively investigated \cite{carpentiero2021distributed,lao2022quantized,Michelusi2022,Nassif2023,sayin2014compressive, xie2013LMS}. Xie and Li \cite{xie2013LMS} design finite level dynamical quantization method for distributed least mean square estimation under finite data rate communications. Sayin and Kozat \cite{sayin2014compressive} propose a single bit diffusion algorithm, which requires least data rate among existing works. Assuming that the Euclidean norm of messages can be transmitted with high precision, Carpentiero et al. \cite{carpentiero2021distributed} and Lao et al. \cite{lao2022quantized} apply the quantizer in \cite{alistarh2017QSGD} and propose adapt-compress-then-combine diffusion algorithm and quantized adapt-then-combine diffusion algorithm, respectively. \MODIFY{The estimates of these algorithms are all mean square bounded, but the almost sure and mean square convergence is not achieved}. 
\MODIFY{
Additionally, the offline distributed estimation problem under finite data rate can be modelled as a distributed learning problem, which is solved by Michelusi et al. \cite{Michelusi2022} and Nassif et al. \cite{Nassif2023}.}
However, under finite data rate communications, how to design an \MODIFY{online} distributed estimation algorithm with estimation errors converging to zero is still an open problem. 


Despite the remarkable progress in distributed estimation under finite data rate communications \cite{carpentiero2021distributed,lao2022quantized,sayin2014compressive, xie2013LMS}, we propose a novel distributed estimation with better effectiveness and lower communication cost. For the effectiveness, the estimates of the algorithm converge to the true value. For the communication cost, the average data rates decay to zero. 

Both of the two issues are challenging. 
For the effectiveness, the main difficulty lies in the selection of consensus protocols to fuse the neighborhood information. Note that consensus protocol is an important part for both the consensus+innovation type distributed estimation algorithms and the diffusion type distributed estimation algorithms. A proper selection of consensus protocols can solve many communication problems in distributed estimation, including the communication cost problem. Under finite data rate communications, there have been many consensus protocols \cite{LiT2011distributed,liu2011distributed,meng2017finite,wang2019consensus,zhao2018consensus}, but many of them have limitations when applied to distributed estimation. For example, the consensus protocol in \cite{zhao2018consensus} requires the states to keep constant in most of the times, which results in a relatively poor effectiveness. Besides, the consensus protocols in \cite{LiT2011distributed,liu2011distributed,meng2017finite,wang2019consensus} are proved to achieve consensus only when all the states are located in known compact sets. This limits their application in the distributed estimation problem due to the randomness of measurements and the lack of \emph{a priori} information on the location of unknown parameter. 

The limitations can be overcome by using the signal-comparison (SC) consensus protocol that we \cite{Ke2023Signal} propose recently. Firstly, the convergence analysis of the SC protocol does not require that all the states are located in known compact sets. Secondly, the SC protocol updates the states at every moment, and therefore achieves a better convergence rate compared with \cite{zhao2018consensus}. Hence, the SC protocol is suitable to be applied in the distributed estimation. 



For the communication cost, if information is transmitted at every moment, the minimum data rate is 1. Therefore, the communication frequency should be reduced to achieve a average data rate that decay to zero.  
The event-triggered strategy is an important method to reduce communication frequency,
\MODIFY{
	and is widely applied in consensus control \cite{YuanS2024JSSC,wu2022dynamic}, distributed Nash equilibrium \cite{Yang2024SCIS} and impulsive synchronization \cite{wang2023event}. For the distributed estimation problem, 
	He et al. \cite{he2022event} propose an event-triggered algorithm where the communication rate can decay to zero at a polynomial rate.} 
However, the mechanism requires accurate transmission of local estimates, making it difficult to extend to the quantized communication case. 
Therefore, it is important to propose a new event-triggered mechanism for the distributed estimation under quantized communications. 

\MODIFY{For the distributed estimation problem under quantized communications}, we propose a new stochastic event-triggered mechanism, which consists of dithering noises and expanding triggering thresholds. The mechanism is suitable for the quantized communication case, because it regards whether the information is transmitted as part of quantized information. 

Based on the SC consensus protocol and the stochastic event-triggered mechanism, we construct the SC-based distributed estimation algorithm. 
The main contributions are summarized as follows.

\begin{enumerate}
	\item For the effectiveness, the estimates of the SC-based algorithm converge to the true value in the almost sure and mean square sense. A polynomial almost sure convergence rate is obtained for the SC-based algorithm. 
	Under finite data rate communications, the SC-based distributed estimation algorithm is the first to achieve convergence. Moreover, it is the first to characterize the almost sure properties of a distributed estimation algorithm under finite data rate communications. 
%
	\item For the communication cost, the average data rates of the SC-based algorithm decay to zero almost surely. The upper bounds of local average data rates are estimated, and both the local and global average data rates converge to zero at a polynomial rate. 
	The SC-based algorithm requires the least average data rates among existing works for distributed estimation \cite{kar2012distributed,Michelusi2022,Nassif2023,sayin2014compressive,xie2013LMS}. 

	\item The trade-off between the convergence rate and the communication cost is established via event-triggered coefficients. 
	A better convergence rate can be achieved by decreasing event-triggered coefficients, while a lower communication cost can be achieved by increasing event-triggered coefficients.  
	The operator of each sensor can decide its own preference on the trade-off by selecting the event-triggered coefficients of adjacent communication channels. 
%
\end{enumerate}

The remainder of the paper is organized as follows. \cref{sec:prob_form} formulates the problem. \cref{sec:algo} introduces the SC consensus protocol and proposes the SC-based distributed estimation algorithm. \cref{sec:conv} analyzes the convergence properties of the algorithm. \cref{sec:comm_cost} calculates the average data rates of the SC-based algorithm to measure the communication cost. 
\cref{sec:trade-off} discusses the trade-off between the convergence rate and the communication cost for the algorithm. 
\cref{sec:simu} gives a simulation example to demonstrate the theoretical results. \cref{sec:concl} concludes the paper. 

\subsection*{Notation}

In the rest of the paper, $ \mathbb{N} $, $\mathbb{R}$, $\mathbb{R}^{n}$, and $ \mathbb{R}^{n\times m} $ are the sets of natural numbers, real numbers, $n$-dimensional real vectors, and $ n \times m $-dimensional real matrices, respectively. $\|x\|$ is the Euclidean norm for vector $ x $, and $\| A \|$ is the induced matrix norm for matrix $ A $. Besides, $ \| x \|_1 $ is the $ L_1 $ norm. 
$ I_n $ is an $n\times n$ identity matrix.  $ \Bone_n $ is the $ n $-dimensional vector whose elements are all ones. $ \diag\{\cdot\} $ denotes the block matrix formed in a diagonal manner
of the corresponding numbers or matrices. $ \col\{\cdot\} $ denotes the column vector stacked by the corresponding numbers or vectors. $ \otimes $ denotes the Kronecker product. \MODIFY{Given two series $ \{a_k\} $ and $ \{b_k\} $, $ a_k=O(b_k) $ means that $ a_k=c_k b_k$ for a bounded $c_k$, and $a_k=o(b_k)$ means that $a_k=c_k b_k$ for a $c_k$ that converges to $ 0 $.
%
}

\section{Problem formulation}\label{sec:prob_form}

This section introduces the graph preliminaries and formulates
the distributed estimation problem under decaying average data rate communications. 

\subsection{Graph preliminaries} 
In this paper, the communications between sensors can be described by an undirected weighted graph $ \mathcal{G} = \left( \mathcal{V}, \mathcal{E}, \mathcal{A} \right) $. $ \mathcal{V} = \{ 1, \ldots, N \} $ is the set of the sensors. $ \mathcal{E} = \{ (i,j) : i,j \in \mathcal{V} \} $ is the edge set. $ (i,j) \in \mathcal{E} $ if and only if the sensor $ i $ and the sensor $ j $ can communicate with each other.  $ \mathcal{A} = (a_{ij})_{N\times N} $ represents the symmetric weighted adjacency matrix of the graph whose elements are all non-negative. $ a_{ij} > 0 $ if and only if $ (i,j)\in \mathcal{E} $. Besides, $ \Neighbour_{i} = \{j : (i,j)\in\mathcal{E}\}$ is used to denote the sensor $ i $'s the neighbor set. Define Laplacian matrix as $ \mL = \mathcal{D} - \mathcal{A} $, where $ \mathcal{D} = \diag\left(\sum_{i\in\Neighbour_1}a_{i1},\ldots,\sum_{i\in\Neighbour_N}a_{iN}\right) $. The graph $ \mathcal{G} $ is said to be connected if $ \text{rank}(\mL) = N-1 $. 

\subsection{Problem statement}
Consider a network $ \mathcal{G} = \left( \mathcal{V}, \mathcal{E}, \mathcal{A} \right) $ with $ N $ sensors. The sensor $ i $ observes the unknown parameter $ \theta \in \mathbb{R}^n $ from the observation model 
\begin{align*}
	\MODIFY{\mathtt{y}_{i,k}} = H_{i,k} \theta + \mathtt{w}_{i,k},
\end{align*}
where \MODIFY{$ k $ is the time index}, $ H_{i,k} \in \mathbb{R}^{m_i\times n} $ is the measurement matrix, $ \mathtt{w}_{i,k} \in \mathbb{R}^{m_i} $ is the observation noise, and $ \mathtt{y}_{i,k} \in \mathbb{R}^{m_i} $ is the observation. Define $ \sigma $-algebra $ \mathcal{F}_{k}^{w} = \sigma(\{\mathtt{w}_{i,t}:i\in \mathcal{V},\ 1\leq t \leq k\}) $. 

The assumptions of the observation model are given as below. 

\begin{assumption}\label{assum:input}
	There exists $ \bar{H} > 0 $ such that $ \Absl{H_{i,k}}\leq \bar{H} $ for all $ k \geq 1 $ and $ i = 1,\ldots,N $.  
	There exists a positive integer $ p $ and a positive real number $ \delta $ such that 
	\begin{equation}\label{condi:exc}
		\frac{1}{p} \sum_{t=k}^{k+p-1} \sum_{i=1}^{N} H_{i,t}^\top H_{i,t} \geq \delta I_n, \ k\geq 1. 
	\end{equation}
\end{assumption}

\begin{remark}
	\MODIFY{The condition \eqref{condi:exc} is the cooperative persistent excitation condition, and is common in existing literature for distributed estimation. For example, \cite{kar2013distributed,sahu2018CIRFE} assumes that $ H_{i,k} $ is constant for all $ k $ and $ \frac{1}{N} \sum_{i=1}^{N} H_{i,k}\tr \Sigma_w^{-1} H_{i,k} $ is invertible, where $ \Sigma_w $ is the nonsingular covariance of $ \mathtt{w}_{i,k} $. This condition is a special case for \cref{assum:input}. }
\end{remark}

\begin{assumption}\label{assum:noise}
	$ \{ \mathtt{w}_{i,k}, \mathcal{F}_k \} $ is a martingale difference sequence such that
	\begin{equation}\label{condi:w_rho}
		\sup_{i\in\mathcal{V},\ k\in\mathbb{N}} \mE\left[ \Abs{\mathtt{w}_{i,k}}^{\rho} \middle| \mathcal{F}^w_{k-1} \right] < \infty,\ \as
	\end{equation} 
	for some \MODIFY{$ \rho > 2 $}. 
\end{assumption}

\begin{remark}
	$ \mathtt{w}_{i,k} $ and $ \mathtt{w}_{j,k} $ is allowed to be correlated for $ i \neq j $, which makes our model applicable to more practical scenarios, such as the distributed target localization \cite{kar2012distributed}. 
\end{remark}

\begin{assumption}\label{assum:connect}
	The communication graph $ \mathcal{G} $ is connected. 
\end{assumption}


The goal of this paper is to cooperatively estimate the unknown parameter $ \theta $. Cooperative estimation requires information exchange between sensors, which brings communication cost. We use the average data rates to describe the communication cost of the distributed estimation. 

\begin{definition}\label{def:bit}
	Given time interval $ [1,k] \cap \mathbb{N} $, the local average data rate for the communication channel where the sensor $ i $ sends messages to the neighbor $ j $ 
	\begin{equation}\label{eq:def_bit}
		\mathtt{B}_{ij}(k) = \frac{\sum_{t=1}^{k} \upzeta_{ij}(t)}{k},
	\end{equation}
	where $ \upzeta_{ij}(t) $ is the bit number of the message that the sensor $ i $ sends to the sensor $ j $ at time $ t $. The global average data rate of communication is
	\begin{align*}
		\mathtt{B}(k) = \frac{\sum_{(i,j)\in\mathcal{E}} \sum_{t=1}^{k} \upzeta_{ij}(t)}{2kM}, 
	\end{align*}
	where $ M $ is the edge number of the communication graph. 
\end{definition}

\begin{remark}
	From \cref{def:bit}, one can get $ \mathtt{B}(k) = \frac{\sum_{(i,j)\in\mathcal{E}} \mathtt{B}_{ij}(k)}{2M} $. 
\end{remark}

\begin{remark}
	The average data rates are used to describe the communication cost because they can represent the consumption of bandwidth, and are also related to transmission energy cost \cite{li2009distributed}. 
\end{remark}

There have been distributed estimation algorithms with $ \mathtt{B}(k) < \infty $. For example, $ \mathtt{B}(k) $ of the distributed least mean square algorithm with $ 2K+1 $ level dynamical quantizer in \cite{xie2013LMS} is $ n\lceil \log_2 (2K+1) \rceil $, where $ \lceil \cdot \rceil $ is the minimum integer that is no smaller than the given number. $ \mathtt{B}(k) $ of the single-bit diffusion algorithm in \cite{sayin2014compressive} is $ 1 $. For effectiveness, these algorithms are shown to be mean square stable \cite{sayin2014compressive,xie2013LMS}. 

Here, we propose a new distributed estimation algorithm with better effectiveness and lower communication cost. For the effectiveness, the estimation errors converge to zero at a polynomial rate. For the communication cost, $ \mathtt{B}_{ij}(k) $ for all communication channels $ (i,j)\in\mathcal{E} $ and $ \mathtt{B}(k) $ also converge to zero. 


\section{Algorithm construction}\label{sec:algo}

The section constructs the distributed estimation algorithm \MODIFY{under the consensus+innovations framework \cite{kar2012distributed}, where a consensus protocol is necessary to fuse the messages transmitted in the network. Therefore,} the SC consensus algorithm \cite{Ke2023Signal} is firstly introduced as the foundation of our distributed estimation algorithm. 

\subsection{The SC consensus protocol \cite{Ke2023Signal}}

In \cite{Ke2023Signal}, we consider the first order multi-agent system 
\begin{equation}\label{sys:MAS}
	\mathtt{x}_{i,k} = \mathtt{x}_{i,k-1} + \mathtt{u}_{i,k},\ \forall i = 1,\ldots,N, 
\end{equation}
where $ \mathtt{x}_{i,k} \in \mathbb{R} $ is the agent $ i $'s state, and $ \mathtt{u}_{i,k} \in \mathbb{R} $ is the input to be designed. The SC consensus protocol for the system \eqref{sys:MAS} is given as in \cref{algo:consensus}.

\begin{algorithm}[H]
	\caption{The SC consensus protocol}
	\label{algo:consensus}
	\begin{algorithmic}
		\STATE \textbf{Input:}
		initial state sequence $ \{x_{i,0}\} $, threshold $ C $, step-size sequence $ \{\alpha_k\} $.
		\STATE \textbf{Output:}
		state sequence $ \{\mathtt{x}_{i,k}\} $.
		\STATE \textbf{for} $ k = 1,2,\ldots, $ \textbf{do}
		\STATE \quad \textbf{Encoding:} The agent $ i $ generates the binary-valued message as 
		\begin{align*}
			\mathtt{s}_{i,k} = \begin{cases}
				1, & \text{if}\ \mathtt{x}_{i,k} + \mathtt{d}_{i,k} < C;\\
				0, & \text{otherwise},
			\end{cases}
		\end{align*}
		where $ \mathtt{d}_{i,k} $ is the noise. 
		\STATE \quad \textbf{Consensus:} The agent $ i $ receives the binary-valued messages $ \mathtt{s}_{j,k} $ for all $ j \in \mathcal{N}_i $, and updates its states by
		\begin{equation}\label{algo:consensus_eq}
			\mathtt{x}_{i,k} = \mathtt{x}_{i,k-1} + \alpha_{k} \sum_{j\in \mathcal{N}_i} a_{ij} \left( \mathtt{s}_{i,k-1} - \mathtt{s}_{j,k-1} \right). 
		\end{equation}
		\STATE \textbf{end for}
	\end{algorithmic}
\end{algorithm}

The effectiveness of \cref{algo:consensus} is analyzed in \cite{Ke2023Signal}. One of the main results is shown below. 

\begin{theorem}[Theorem 1 of \cite{Ke2023Signal}]\label{thm:consensus}
	Assume that the communication graph is connected, $ \sum_{k=1}^{\infty} \alpha_k = \infty $, $ \sum_{k=1}^{\infty} \alpha_k^2 < \infty $, and the noise sequence $ \{\mathtt{d}_{i,k}\} $ is independent and identically distributed (i.i.d.) with a strictly increasing distribution function $ F(\cdot) $. Then, for \cref{algo:consensus}, we have $ \lim_{k\to\infty} \mathtt{x}_{i,k} = \frac{1}{N}\sum_{j=1}^{N} x_{j,0} $ almost surely. 
\end{theorem}

\begin{remark}
	\cref{thm:consensus} shows that \cref{algo:consensus} can achieve the almost sure consensus. 
	Therefore,  \cref{algo:consensus} can be used to solve the information transmission problem of distributed identification under binary-valued communications. 
\end{remark}

\begin{remark}
	\MODIFY{The design idea of \cref{algo:consensus} is based on the comparison of the binary-valued messages $ \mathtt{s}_{i,k} $ and $ \mathtt{s}_{j,k} $. If $ \mathtt{s}_{i,k} - \mathtt{s}_{j,k} = 1 $, then $ \mathtt{s}_{i,k} = 1 $ and $ \mathtt{s}_{j,k} = 0 $. From the distributions of $ \mathtt{s}_{i,k} $ and $ \mathtt{s}_{j,k} $, one can get that $ \mathtt{x}_{i,k} $ is more likely to be less than $ \mathtt{x}_{j,k} $. Therefore, in \cref{algo:consensus}, $ \mathtt{x}_{i,k} $ increases, and $ \mathtt{x}_{j,k} $ decreases. Conversely, if $ \mathtt{s}_{i,k} - \mathtt{s}_{j,k} = -1 $, then $ \mathtt{x}_{i,k} $ decreases, and $ \mathtt{x}_{j,k} $ increases.}
\end{remark}

\MODIFY{
\begin{remark}\label{remark:dither}
	The noise $ \mathtt{d}_{i,k} $ with strictly increasing distribution function is necessary for \cref{algo:consensus}. Without such a noise, the states $ \mathtt{x}_{i,k} $ will keep constant if all the states are greater (or smaller) than the threshold $ C $, and hence, the consensus may not be achieved. With the noise $ \mathtt{d}_{i,k} $, $ \mE\left[ \mathtt{s}_{i,k} \middle| \mathtt{x}_{i,k} \right] $ is strictly decreasing with $ \mathtt{x}_{i,k} $. Therefore, when $  \mathtt{x}_{i,k} \neq  \mathtt{x}_{j,k} $, the stochastic properties of $ \mathtt{s}_{i,k} $ and $ \mathtt{s}_{j,k} $ are different even if $ \mathtt{x}_{i,k} $ and $ \mathtt{x}_{j,k} $ are all greater (or smaller) than the threshold $ C $. The consensus can be thereby achieved. 
%
\end{remark}
}


\subsection{The SC-based distributed estimation algorithm}
The subsection propose the SC-based distributed estimation algorithm in \cref{algo:DE}.

\begin{algorithm}[H]
	\caption{The SC-based distributed estimation algorithm.}
	\label{algo:DE}
	\begin{algorithmic}
		\STATE \textbf{Input:}
		initial estimate sequence $ \{\hat{\theta}_{i,0}\} $, event-triggered coefficient sequence $ \{\nu_{ij}\} $ with $ \nu_{ij} = \nu_{ji} \geq 0 $, noise coefficient sequence $ \{b_{ij}\} $ with $ b_{ij} = b_{ji} > 0 $, step-size sequences $ \{\alpha_{ij,k}\} $ with $ \alpha_{ij,k} = \alpha_{ji,k} > 0 $ and $ \{\beta_{i,k}\} $ with $ \beta_{i,k} > 0 $.
		\STATE \textbf{Output:}
		estimate sequence $ \{\hat{\uptheta}_{i,k}\} $.
		\STATE \textbf{for} $ k = 1,2,\ldots, $ \textbf{do}
		\STATE \quad \textbf{Compressing:} If $ k = n q+l $ for some $ q \in \mathbb{N} $ and $ l \in \{1,\ldots,n\} $, then the sensor $ i $ generates $ \varphi_k $ as the $ n $-dimensional vector whose $ l $-th element is $ 1 $ and the others are $ 0 $. The sensor $ i $ uses $ \varphi_{k} $ to compress the previous local estimate $ \hat{\uptheta}_{i,k-1} $ into the scalar $ \mathtt{x}_{i,k} = \varphi_{k}\tr \hat{\uptheta}_{i,k-1} $. 
		\STATE \quad \textbf{Encoding:} The sensor $ i $ generates the dithering noise $ \mathtt{d}_{i,k} $ with Laplacian distribution $ Lap(0,1) $. Then, the sensor $ i $ generates the binary-valued message for the neighbor $ j $
		\begin{align*}
			\mathtt{s}_{ij,k} = \begin{cases}
				1, & \text{if}\ \mathtt{x}_{i,k} + b_{ij} \mathtt{d}_{i,k} > 0;\\
				-1, & \text{otherwise}.
			\end{cases}
		\end{align*}
		\STATE \quad \textbf{Data Transmission:} Set $ C_{ij,k} = \nu_{ij} b_{ij} \ln k $. If $ \absl{\mathtt{x}_{i,k} + b_{ij} \mathtt{d}_{i,k}} > C_{ij,k} $, then the sensor $ i $ sends the 1 bit message \MODIFY{$ \mathtt{s}_{ij,k} $} to the neighbor $ j $. Otherwise, the sensor $ i $ does not send any message to the neighbor $ j $. 
		\STATE \quad \textbf{Data Receiving:} If the sensor $ i $ receives 1 bit message $ \mathtt{s}_{ji,k} $ from its neighbor $ j $, then set $ \hat{\mathtt{s}}_{ji,k} = \mathtt{s}_{ji,k} $. Otherwise, set $ \hat{\mathtt{s}}_{ji,k} = 0 $. 
		\STATE \quad \textbf{Information fusion:} Apply the modified \cref{algo:consensus} to fuse the neighborhood information.
		 \begin{equation}\label{algo:DE_consensus}
		 	\check{\uptheta}_{i,k} = \hat{\uptheta}_{i,k-1} + \varphi_k \sum_{j\in \mathcal{N}_i} \alpha_{ij,k} a_{ij} \left( \hat{\mathtt{s}}_{ji,k} - G_{ij,k} (\mathtt{x}_{i,k}) \right)
		 \end{equation}
		 where $ G_{ij,k}(x) = F((x - C_{ij,k})/b_{ij}) - F((- x - C_{ij,k})/b_{ij}) $, and $ F(\cdot) $ is the distribution function of $ Lap(0,1) $. 
		 \STATE \quad \textbf{Estimate update:} Use the observation $ \mathtt{y}_{i,k} $ to update the local estimate.
		 \begin{equation}
		 	\hat{\uptheta}_{i,k} = \check{\uptheta}_{i,k} + \beta_{i,k}H_{i,k}\tr \left(\mathtt{y}_{i,k}-H_{i,k} \hat{\uptheta}_{i,k-1}\right). 
		 \end{equation}
		 \STATE \textbf{end for}
	\end{algorithmic}
\end{algorithm}

In \cref{algo:DE}, dithering noise $ \mathtt{d}_{i,k} $ is used for the encoding step and the event-triggered condition. The independence assumption for $ \mathtt{d}_{i,k} $ is required. 

\begin{assumption}\label{assum:dither}
	$ \mathtt{d}_{i,k} $ and $ \mathtt{d}_{j,t} $ are independent when $ k \neq t $ or $ i \neq j $. And, $ \mathtt{d}_{i,k} $ and $ \mathtt{w}_{j,t} $ are independent for all $ i, j\in \mathcal{V} $ and $ k,t \in \mathbb{N} $.
\end{assumption}

Following remarks are given for \cref{algo:DE}. 

\begin{remark}
	\MODIFY{
		The requirement that $ \alpha_{ij,k} = \alpha_{ji,k} $ in \cref{algo:DE} is weak among existing literature. In the distributed estimation algorithms in \cite{kar2013distributed,kar2012distributed,Jakovetic2023heavy-tail,wang2023differentially}, it is required that $ \alpha_{ij,k} = \alpha_{i^\prime j^\prime,k} $ for all $ (i,j), (i^\prime,j^\prime) \in \mathcal{E} $. He et al. \cite{he2022event} and Zhang and Zhang \cite{ZhangQ2012} relax this condition, but still require that $ \lim_{k\to\infty} \frac{\alpha_{ij,k}}{\alpha_{i^\prime j^\prime,k}} = 1 $ for all $ (i,j), (i^\prime,j^\prime) \in \mathcal{E} $, and hence the step-sizes $ \alpha_{ij,k} $ converge to $ 0 $ with the same order. For comparison, in \cref{algo:DE},  $ \alpha_{ij,k} = \alpha_{i^\prime j^\prime,k} $ is required only when $ i = j^\prime $ and $ j = i^\prime $, which is more easily implemented since it only requires the communication between adjacent sensors $ i $ and $ j $, and the step-sizes $ \alpha_{ij,k} $ in \cref{algo:DE} are allowed to converge to $ 0 $ with different orders. Here, we give one of the techniques to achieve $ \alpha_{ij,k} = \alpha_{ji,k} $, which is a two-step protocol before running \cref{algo:DE}. Firstly, the operators of the sensors $ i $ and $ j $ select positive numbers $ \bar{\alpha}_{ij,1} $, $ \bar{\gamma}_{ij} $ and $ \bar{\alpha}_{ji,1} $, $ \bar{\gamma}_{ji} $, respectively, and then transmit the selected numbers to each other. Secondly, set $ \alpha_{ij,k} = \alpha_{ji,k} = \frac{\alpha_{ij,1}}{k^{\gamma_{ij}}} $, where $ \alpha_{ij,1} = \frac{\bar{\alpha}_{ij,1}+\bar{\alpha}_{ji,1}}{2} $ and $ \gamma_{ij} = \frac{\bar{\gamma}_{ij}+\bar{\gamma}_{ji}}{2} $. By using this technique, it requires only finite bits of communications to achieve $ \alpha_{ij,k} = \alpha_{ji,k} $ if $ \bar{m} \bar{\alpha}_{ij,1} $, $ \bar{m} \bar{\gamma}_{ij} $, $ \bar{m} \bar{\alpha}_{ji,1} $, $ \bar{m} \bar{\gamma}_{ji} $ are all integers for some positive $ \bar{m} $. Similar techniques can be applied to achieve $ \nu_{ij} = \nu_{ji} $ and $ b_{ij} = b_{ji} $ in \cref{algo:DE}. 
}
\end{remark}

\begin{remark}
	A new stochastic event-triggered mechanism is applied to \cref{algo:DE}. The main idea is to use the dithering noises and the expanding triggering thresholds. When $ \nu_{ij} > 0 $, the threshold $ C_{ij,k} $ goes to infinity. 
	Hence, the probability that $ \abs{\mathtt{x}_{i,k} + b_{ij}\mathtt{d}_{i,k}} > C_{ij,k} $ decays to zero, which implies that the communication frequency is reduced. 
%
\end{remark}

\begin{remark}
	The stochastic event-triggered mechanism used in \cref{algo:DE} is significantly different from existing ones. When the information is not transmitted at a certain moment, the traditional event-triggered mechanisms \cite{he2022event} use the recently received message as an approximation of the untransmitted message. Note that in the binary-valued communication case, $ 1 $ and $ -1 $ represent opposite information. Then, in this case, approximation technique of \cite{he2022event} can only be used when the recently received message is the same as the untransmitted message. This constraint makes it difficult to reduce communication frequency to zero through event-triggered mechanisms. 
	To overcome the difficulty, a new approximation method is used in \cref{algo:DE}. When the information is not transmitted at a certain moment, our stochastic event-triggered mechanism uses $ 0 $ as an approximation of the untransmitted information. The approximation technique expands the binary-valued message $ \mathtt{s}_{ji,k} $ to triple-valued message $ \hat{\mathtt{s}}_{ji,k} $. The message $ \hat{\mathtt{s}}_{ji,k} $ contains information on whether $ \mathtt{s}_{ji,k} $ is transmitted or not. Hence, the statistical properties of whether $ \mathtt{s}_{ji,k} $ is transmitted can be better utilized.
\end{remark}

\begin{remark}
	\MODIFY{In \cref{algo:DE}, the dithering noise $ \mathtt{d}_{i,k} $ is artificial, and generated under a given distribution function. The necessity of introducing $ \mathtt{d}_{i,k} $ is similar to that in \cref{algo:consensus}, which has been explained in \cref{remark:dither}. For similar reasons, dithering noises are often used to avoid the influence of quantization error \cite{aysal2008distributed,Gustafsson2013generating,WangLY2003System}.}
	Besides, in \cref{algo:DE}, the dithering noise $ \mathtt{d}_{i,k} $ is not necessarily Laplacian distributed. $ \mathtt{d}_{i,k} $ can be any other types with continuous and strictly increasing distribution $ F(\cdot) $, including Gaussian noises and the heavy-tailed noises \cite{Jakovetic2023heavy-tail}. 
	For the polynomial decaying rate of $ \mathtt{B}(k) $, the triggering threshold $ C_{ij,k} $ can be changed accordingly. 
\end{remark}

\begin{remark}\label{remark:replace_DE}
\MODIFY{
	In \eqref{algo:DE_consensus}, we use $ G_{ij,k}(\mathtt{x}_{i,k}) $ to replace $ \hat{\mathtt{s}}_{ij,k} $ in order to reduce the variances of the estimates, because $ \mE\left[ \hat{\mathtt{s}}_{ij,k}\middle| \mathcal{F}_{k-1} \right] = G_{ij,k}(\mathtt{x}_{i,k}) $, where $ \mathcal{F}_{k} = \sigma(\{\mathtt{w}_{i,t},\mathtt{d}_{i,t}:i=1,\ldots,N,\ 1\leq t \leq k\}) $. 
}
%
\end{remark}

\section{Convergence analysis}\label{sec:conv}

The convergence properties of \cref{algo:DE} is analyzed in this section. The almost sure convergence and mean square convergence are obtained in \cref{ssec:conv}. Then, the almost sure convergence rate is calculated in \cref{ssec:conv_rate}. 
\subsection{Convergence}\label{ssec:conv}

\MODIFY{This subsection focuses on the almost sure and mean square convergence of \cref{algo:DE}. The following theorem gives a new step-size condition, where the step-sizes are allowed to converge to zero with different orders, and the estimates of \cref{algo:DE} are proved to converge to the true value almost surely. }

\begin{theorem}\label{thm:as_conv}
	Suppose the step-size sequences $ \{\alpha_{ij,k}\} $ and $ \{\beta_{i,k}\} $ satisfy
	\begin{enumerate}[label={\roman*)}]
		\item $ \sum_{k=1}^{\infty} \alpha^2_{ij,k} < \infty $ and $ 
		\alpha_{ij,k+1} = O\left( \alpha_{ij,k} \right) $ for all $ (i,j)\in\mathcal{E} $;
		\item $ \sum_{k=1}^{\infty} \beta^2_{i,k} < \infty $ and $
		\beta_{i,k+1} = O\left( \beta_{i,k} \right) $ for all $ \forall i\in\mathcal{V} $; 
		\item $ \sum_{k=1}^{\infty} z_k = \infty $ for $ z_k = \min\left\{ \frac{\alpha_{ij,k}}{k^{\nu_{ij}}}, (i,j)\in\mathcal{E}; \beta_{i,k}, i\in\mathcal{V} \right\} $.
	\end{enumerate}
	Then, under \cref{assum:input,assum:noise,assum:connect,assum:dither}, the estimate $ \hat{\uptheta}_{i,k} $ in \cref{algo:DE} converges to the true value $ \theta $ almost surely. 
\end{theorem}

\begin{proof}
	By $ \mE\left[ \hat{\mathtt{s}}_{ji,k}\middle| \mathcal{F}_{k-1} \right] = G_{ji,k}(\mathtt{x}_{j,k}) $, one can get
	\begin{align}\label{eq:condi_E}
		& \mE\left[\left(\hat{\mathtt{s}}_{ji,k}-G_{ji,k}(\mathtt{x}_{j,k-1})\right)^2\middle|\mathcal{F}_{k-1}\right] \\
		= & \mE \left[ \hat{\mathtt{s}}_{ji,k}^2 \middle|\mathcal{F}_{k-1} \right] - G_{ij,k}^2(\mathtt{x}_{j,k}) 
		\nonumber\\
		= & F((\mathtt{x}_{j,k}-C_{ji,k})/b_{ji})+F((-\mathtt{x}_{j,k}-C_{ji,k})/b_{ji})-G_{ji,k}^2(\mathtt{x}_{j,k}), \nonumber
	\end{align}
	where the $ \sigma $-algebra $ \mathcal{F}_{k-1} $ is defined in \cref{remark:replace_DE}.
	Besides by the Lagrange mean value theorem \cite{zorich}, given $ (i,j) \in \mathcal{E} $, there exists $ \upxi_{ij,k} $ between $ \mathtt{x}_{i,k} $ and $ \mathtt{x}_{j,k} $ such that 
	\begin{align*}
		G_{ji,k}(\mathtt{x}_{j,k}) - G_{ij,k}(\mathtt{x}_{i,k}) = g_{ij,k}(\upxi_{ij,k})\left(\mathtt{x}_{j,k}-\mathtt{x}_{i,k}\right),
	\end{align*}
	where 
	\begin{align*}
		g_{ij,k}(x) = g_{ji,k}(x) = \left.\left(f\left( \frac{x-C_{ij,k}}{b_{ij}} \right) + f\left( \frac{-x-C_{ij,k}}{b_{ij}} \right)\right) \middle/ b_{ij} \right.,
	\end{align*}
	and $ f(\cdot) $ is the density function of $ Lap(0,1) $. 
	Denote $ \tilde{\uptheta}_{i,k} = \hat{\uptheta}_{i,k} - \theta $.
	Then, it holds that
	\begin{align*}
		\mE \left[ \Absl{\tilde{\uptheta}_{i,k}}^2 \middle| \mathcal{F}_{k-1} \right] 
		= &\Absl{\tilde{\uptheta}_{i,k-1}}^2
		- 2\beta_{i,k} \left(H_{i,k} \tilde{\uptheta}_{i,k-1}\right)^2 \\
		& + 2 \varphi_k\tr \tilde{\uptheta}_{i,k-1} \sum_{j\in \mathcal{N}_i}\alpha_{ij,k} a_{ij} g_{ij,k}(\upxi_{ij,k})\left(\mathtt{x}_{j,k}-\mathtt{x}_{i,k}\right) \nonumber\\
		&+ O\left( \beta^2_{i,k}\left(\Absl{\tilde{\uptheta}_{i,k-1}}^2+1\right) + \sum_{j \in \mathcal{N}_i} \alpha^2_{ij,k} \right). \nonumber
	\end{align*}
	Denote $ \tilde{\mathtt{x}}_{i,k} = \varphi_k\tr \tilde{\uptheta}_{i,k-1} = \mathtt{x}_{i,k} - \varphi_{k}\tr \theta $ and $ \tilde{\mathtt{X}}_k = [\tilde{\mathtt{x}}_{1,k}, \ldots, \tilde{\mathtt{x}}_{N,k}]\tr $. Then, one can get
	\begin{align}\label{eq:LG}
		& \sum_{i=1}^{N} 2 \varphi_k\tr \tilde{\uptheta}_{i,k-1} \sum_{j\in \mathcal{N}_i}\alpha_{ij,k} a_{ij} g_{ij,k}(\upxi_{ij,k})\left(\mathtt{x}_{j,k}-\mathtt{x}_{i,k}\right) \\
		= & \sum_{i=1}^{N} 2 \tilde{\mathtt{x}}_{i,k} \sum_{j\in \mathcal{N}_i}\alpha_{ij,k} a_{ij} g_{ij,k}(\upxi_{ij,k})\left(\tilde{\mathtt{x}}_{j,k}-\tilde{\mathtt{x}}_{i,k}\right) 
		= - 2 \tilde{\mathtt{X}}_k\tr \mathtt{L}_{G,k} \tilde{\mathtt{X}}_k, \nonumber
	\end{align}
	where $ \mathtt{L}_{G,k} = (\mathtt{l}^G_{ij,k})_{N\times N} $ is a Laplacian matrix with $ \mathtt{l}^G_{ii,k} = \sum_{j\in\mathcal{N}_i} \alpha_{ij,k} a_{ij} g_{ij,k}(\upxi_{ij,k}) $ and
	$ \mathtt{l}^G_{ij,k} = -\alpha_{ij,k} a_{ij} g_{ij,k}(\upxi_{ij,k}) $ for $ i \neq j $. 
	Therefore, we have
	\begin{align}\label{eq:CondE_sumtheta}
		\mE \left[ \sum_{i=1}^{N} \Absl{\tilde{\uptheta}_{i,k}}^2 \middle| \mathcal{F}_{k-1} \right]
		= & \sum_{i=1}^{N} \Absl{\tilde{\uptheta}_{i,k-1}}^2
		- 2 \sum_{i=1}^{N} \beta_{i,k} \left(H_{i,k} \tilde{\uptheta}_{i,k-1}\right)^2 - 2 \tilde{\mathtt{X}}_k\tr \mathtt{L}_{G,k} \tilde{\mathtt{X}}_k \\
		&  + O\left( \sum_{i=1}^{N} \beta^2_{i,k}\left(\Absl{\tilde{\uptheta}_{i,k-1}}^2+1\right) + \sum_{(i,j)\in\mathcal{E}} \alpha^2_{ij,k} \right). \nonumber
	\end{align}
	Then, by Theorem 1.3.2 of \cite{GuoL2020}, $ \sum_{i=1}^{N} \Absl{\tilde{\uptheta}_{i,k}}^2 $ converges to a finite value almost surely, and
	\begin{equation}\label{ineq:sum_H+L<inf}
		\sum_{k=1}^\infty \left( \sum_{i=1}^{N}  \beta_{i,k} \left(H_{i,k} \tilde{\uptheta}_{i,k-1}\right)^2 + \tilde{\mathtt{X}}_k\tr \mathtt{L}_{G,k} \tilde{\mathtt{X}}_k \right) < \infty,\ \as, 
	\end{equation} 
	
	By the convergence of $ \sum_{i=1}^{N} \Absl{\tilde{\uptheta}_{i,k}}^2 $, 
	$ \tilde{\mathtt{x}}_{i,k} = \varphi_{k}\tr \tilde{\uptheta}_{i,k} $ is uniformly bounded almost surely. Then, by \cref{lemma:inf_kf} in \cref{appen}, it holds that 
	\begin{equation}\label{eq:def_udg}
		\underline{\mathtt{g}} := \inf_{(i,j)\in\mathcal{E}, k\in\mathbb{N}} k^{\nu_{ij}} g_{ij,k}(\upxi_{ij,k}) > 0,\ \as
	\end{equation}
	Hence, one can get 
	\begin{equation}\label{eq:Laplacian_lower}
		\mathtt{L}_{G,k} \geq \left(\min_{(i,j)\in\mathcal{E}} \frac{\alpha_{ij,k}}{k^{\nu_{ij}}}\right) \underline{\mathtt{g}} \lambda_2(\mathcal{L})\left( I_N - J_N \right),
	\end{equation}
	where $ \lambda_2(\mathcal{L}) $ is the second smallest eigenvalue of $ \mathcal{L} $, and $ J_N = \frac{1}{N}\Bone_N\tr \Bone_N $. 
	
	Denote 
	\begin{align*}
		\tilde{\Theta}_k &= \col\{\tilde{\uptheta}_{1,k},\ldots,\tilde{\uptheta}_{N,k}\},\ 
		\mathbb{H}_k = \diag\{ H_{1,k}\tr H_{1,k},\ldots, H_{N,k}\tr H_{N,k}\}, \\
		\mathbb{H}_{\beta,k} &= \diag\{\beta_{1,k} H_{1,k}\tr H_{1,k},\ldots, \beta_{N,k} H_{N,k}\tr H_{N,k}\},\\ 
		\Phi_k &= \mathbb{H}_k + \underline{\mathtt{g}} \lambda_2(\mathcal{L})\left( I_N - J_N \right) \otimes \varphi_k \varphi_k\tr, \\
		\mathtt{W}_k &= \col\{\beta_{1,k} H_{1,k}\tr \mathtt{w}_{1,k}, \ldots, \beta_{N,k} H_{N,k}\tr \mathtt{w}_{N,k}\},\\
		& \quad\  + \left[\left(\varphi_{k}\sum_{j\in \mathcal{N}_1}\alpha_{1j,k} a_{1j}(\hat{\mathtt{s}}_{j1,k}-G_{j1,k}(\mathtt{x}_{j,k}))\right)\tr,\ldots,\right.\\
		&\qquad\qquad\qquad\qquad\left.\left(\varphi_{k}\sum_{j\in \mathcal{N}_N}\alpha_{Nj,k} a_{Nj} a_{Nj}(\hat{\mathtt{s}}_{jN,k}-G_{jN,k}(\mathtt{x}_{j,k}))\right)\tr\right]\tr.  
	\end{align*}
	Then, $ \mathtt{W}_k $ is $ \mathcal{F}_{k} $-measurable, and
	\begin{align}
		& \tilde{\Theta}_k =\left( I_{N\times n} - \mathbb{H}_{\beta,k} - \mathtt{L}_{G,k}\otimes \varphi_k\varphi_k\tr \right) \tilde{\Theta}_{k-1} + \mathtt{W}_k, \label{eq:recur_Theta1}\\
		&\mE \left[\mathtt{W}_k\middle|\mathcal{F}_{k-1}\right] = 0, \ 
		\mE \left[ \Absl{\mathtt{W}_k}^2 \middle|\mathcal{F}_{k-1} \right] = O\left( \sum_{i=1}^{N} \beta_{i,k}^2 + \sum_{(i,j)\in\mathcal{E}} \alpha_{ij,k}^2 \right). \nonumber
	\end{align}
	By the almost sure uniform boundedness of $ \tilde{\Theta}_k $ and \eqref{eq:recur_Theta1}, one can get
	\begin{align}\label{eq:minus_Theta}
		\mathtt{P}_k := \frac{\tilde{\Theta}_{k} - \tilde{\Theta}_{k-1} -\mathtt{W}_k}{\sum_{i=1}^{N} \beta_{i,k} + \sum_{(i,j)\in\mathcal{E}} \alpha_{ij,k}}
	\end{align}
	is $ \mathcal{F}_{k-1} $-measurable and almost surely uniformly bounded. 
	By \eqref{eq:Laplacian_lower}, it holds that
	\begin{align}
		\sum_{i=1}^{N}  \beta_{i,k} \left(H_{i,k} \tilde{\uptheta}_{i,k-1}\right)^2 + \tilde{\mathtt{X}}_k\tr \mathtt{L}_{G,k} \tilde{\mathtt{X}}_k
		\geq z_k \tilde{\Theta}_k\tr \Phi_k \tilde{\Theta}_k. \label{ineq:z*Theta*Phi*Theta} 
	\end{align}
	Besides by Lemma 5.4 of \cite{xie2018analysis}, there exists $ \underline{\mathtt{H}} > 0 $ almost surely such that 
	\begin{align}\label{ineq:underline_H}
		\sum_{t=k-np+1}^{k} \Phi_t
		= \sum_{t=k-np+1}^{k} \mathbb{H}_t + \underline{\mathtt{g}} \lambda_2(\mathcal{L})p \left( I_N - J_N \right) \otimes I_n
		\geq\underline{\mathtt{H}}. 
	\end{align}

	By \eqref{eq:minus_Theta}, 
	one can get
	\begin{align}\label{eq:minus_z*Theta*Phi*Theta}
		 & \sum_{t=npr+1}^{npr+np} z_t \tilde{\Theta}_{pr+p}\tr \Phi_t \tilde{\Theta}_{npr+np} - \sum_{t=npr+1}^{npr+np} z_t \tilde{\Theta}_t\tr \Phi_t \tilde{\Theta}_t  \\
		 = & \sum_{t=npr+1}^{npr+np} z_t \sum_{l=t+1}^{npr+np} \left( \tilde{\Theta}_l\tr \Phi_t \tilde{\Theta}_l - \tilde{\Theta}_{l-1}\tr \Phi_t \tilde{\Theta}_{l-1} \right) \cr
		 = &  \sum_{t=npr+1}^{npr+np} z_t \sum_{l=t+1}^{npr+np} \left( 2 \mathtt{W}_l\tr \Phi_t \tilde{\Theta}_{l-1} + \mathtt{W}_l\tr \Phi_t \mathtt{W}_l  \right) \nonumber\\
		 & + O\left( \sum_{t=npr+1}^{npr+np} z_t \left(\sum_{l=t+1}^{npr+np} \sum_{i=1}^{N} \beta_{i,l} + \sum_{l=t+1}^{npr+np} \sum_{(i,j)\in\mathcal{E}} \alpha_{ij,l}\right) \right) \cr
		 & + \sum_{t=npr+1}^{npr+np} 2z_t \left(\sum_{l=t+1}^{npr+np} \sum_{i=1}^{N} \beta_{i,l} \mathtt{W}_l\tr \Phi_t \mathtt{P}_l +  \sum_{l=t+1}^{npr+np} \sum_{(i,j)\in\mathcal{E}} \alpha_{ij,l} \mathtt{W}_l\tr \Phi_t \mathtt{P}_l\right),\ \as \nonumber
	\end{align}
	By $ \sum_{k=1}^{\infty} \alpha^2_{ij,k} < \infty $ and $ \sum_{k=1}^{\infty} \beta^2_{i,k} < \infty  $, we have
	\begin{align*}
		\sum_{r=1}^{\infty} \sum_{t=npr+1}^{npr+np} z_t \left(\sum_{l=t+1}^{npr+np} \sum_{i=1}^{N} \beta_{i,l} +  \sum_{l=t+1}^{npr+np} \sum_{(i,j)\in\mathcal{E}} \alpha_{ij,l}\right) < \infty. 
	\end{align*}
	By Theorem 1.3.10 of \cite{GuoL2020}, one can get
	\begin{gather*}
		\sum_{r=1}^{\infty} \sum_{t=npr+1}^{npr+np} \sum_{l=t+1}^{npr+np} 2 z_t  \mathtt{W}_l\tr \Phi_t \tilde{\Theta}_{l-1} < \infty,\ \as, \\
		\sum_{r=1}^{\infty} \sum_{t=npr+1}^{npr+np} \sum_{l=t+1}^{npr+np} 2 z_t \left(\sum_{i=1}^{N} \beta_{i,l} \mathtt{W}_l\tr \Phi_t \mathtt{P}_l + \sum_{(i,j)\in\mathcal{E}} \alpha_{ij,l} \mathtt{W}_l\tr \Phi_t \mathtt{P}_l\right) < \infty, \ \as
	\end{gather*}
	\MODIFY{By Theorem 1.3.9 of \cite{GuoL2020} with $ \alpha = 1 $, we have
	\begin{align*}
		\sum_{r=1}^{\infty} \sum_{t=npr+1}^{npr+np} \sum_{l=t+1}^{npr+np} z_t \mathbb{E} \Absl{\mathtt{W}_l}^2 \cdot 
		\frac{1}{\mathbb{E} \Absl{\mathtt{W}_l}^2} \left(  \mathtt{W}_l\tr \Phi_t \mathtt{W}_l - \mE\left[ \mathtt{W}_l\tr \Phi_t \mathtt{W}_l \middle| \mathcal{F}_{l-1} \right] \right) < \infty,\ \as
	\end{align*}}
	Besides, $ \mE\left[ \mathtt{W}_l\tr \Phi_t \mathtt{W}_l \middle| \mathcal{F}_{l-1} \right] = O\left(\! \left( \sum_{i=1}^{N} \beta_{i,l} + \sum_{(i,j)\in\mathcal{E}} \alpha_{ij,l}  \right)^2 \right) $ almost surely. Then, 
	\begin{align*}
		\sum_{r=1}^{\infty} \sum_{t=npr+1}^{npr+np} z_t \sum_{l=t+1}^{npr+np} \mE\left[ \mathtt{W}_l\tr \Phi_t \mathtt{W}_l \middle| \mathcal{F}_{l-1} \right] < \infty,\ \as
	\end{align*}
	Therefore by \eqref{ineq:sum_H+L<inf}, \eqref{ineq:z*Theta*Phi*Theta}-\eqref{eq:minus_z*Theta*Phi*Theta}, we have
	\begin{align*}
		& \underline{\mathtt{H}} \sum_{r=1}^{\infty} \left(\min_{npr+1 \leq t \leq npr+np} z_t\right) \Absl{\tilde{\Theta}_{npr+np}}^2 \\
		\leq & \sum_{r=1}^{\infty} \left(\min_{npr+1 \leq t \leq npr+np} z_t\right) \tilde{\Theta}_{npr+np}\tr \left( \sum_{t=npr+1}^{npr+np}\Phi_t\right) \tilde{\Theta}_{npr+np} \\
		\leq & \sum_{r=1}^{\infty} \sum_{t=npr+1}^{npr+np} z_t \tilde{\Theta}_{npr+np}\tr \Phi_t \tilde{\Theta}_{npr+np} 
		= \sum_{k=1}^{\infty} z_k \tilde{\Theta}_{k}\tr \Phi_k \tilde{\Theta}_{k} + O(1) \\
		\leq & \sum_{k=1}^{\infty} \left(\sum_{i=1}^{N}  \beta_{i,k} \left(H_{i,k} \tilde{\uptheta}_{i,k-1}\right)^2 + \tilde{\mathtt{X}}_k\tr \mathtt{L}_{G,k} \tilde{\mathtt{X}}_k\right) + O(1 )
		 < \infty,\ \as
	\end{align*}
	Then, by \cref{lemma:sum_min_z} in \cref{appen}, there exist $ \mathtt{k}_1 < \mathtt{k}_2 <\cdots $  such that $ \lim\limits_{t\to\infty}\Absl{\tilde{\Theta}_{\mathtt{k}_t}}^2 $  $= 0 $ almost surely. Note that $ \sum_{i=1}^{N} \Absl{\tilde{\uptheta}_{i,k}}^2 = \Absl{\tilde{\Theta}_{k}}^2 $ converges to a finite value. Then, the value is $ 0 $, which proves the theorem. 
\end{proof}

\MODIFY{
	\begin{remark}
		The estimates of \cref{algo:DE} can converge to the true value because the algorithm is designed by using the idea of stochastic approximation \cite{ChenHF2003SA}. In \cref{algo:DE}, $ \hat{\mathtt{s}}_{ji,k} - G_{ij,k} (\mathtt{x}_{i,k}) = G_{ij,k} (\mathtt{x}_{j,k}) - G_{ij,k} (\mathtt{x}_{i,k}) + \hat{\mathtt{s}}_{ji,k} - G_{ij,k} (\mathtt{x}_{j,k}) $ and $ \mathtt{y}_{i,k}-H_{i,k} \hat{\uptheta}_{i,k-1} = -H_{i,k} \tilde{\uptheta}_{i,k-1} + \mathtt{w}_{i,k} $, where $ \hat{\mathtt{s}}_{ji,k} - G_{ij,k} (\mathtt{x}_{j,k}) $ and $ \mathtt{w}_{i,k} $ are martingale difference with bounded variance, and  
		\begin{gather*}
			G_{ij,k} (\varphi_k\tr \hat{\theta}_{j}) - G_{ij,k} (\varphi_k\tr \hat{\theta}_{i}) = 0,\ \forall (i,j)\in\mathcal{E}, k\in\mathbb{N}; \ 
			H_{i,k} (\hat{\theta}_{i}-\theta) = 0, \forall i\in\mathcal{V}, k\in\mathbb{N}
		\end{gather*}
		holds if and only if $ \hat{\theta}_{i} = \theta $ for all $ i $. Besides, under i) and ii) of \cref{thm:as_conv}, the step-sizes converge to $ 0 $. These algorithm characteristics based on stochastic approximation enable the estimates to converge to the true value \cite{ChenHF2003SA}.
	\end{remark}
}

\begin{remark}
	If $ \alpha_{ij,k} $ and $ \beta_{i,k} $ are all polynomial, iii) of \cref{thm:as_conv} is equivalent to $ \sum_{k=1}^{\infty} \frac{\alpha_{ij,k}}{k^{\nu_{ij}}} = \infty $ for all $ (i,j) \in \mathcal{E} $ and $ \sum_{k=1}^{\infty} \beta_{i,k} = \infty $ for all $ i \in \mathcal{V} $. Under this case, the step-sizes can be designed in a distributed manner. 
\end{remark}

\begin{remark}
	Note that 
	$ 2\sum_{t=1}^{k} \frac{\alpha_{ij,k}}{k^{\nu_{ij}}} \leq \sum_{t=1}^{k} \alpha_{ij,k}^2 + \sum_{t=1}^{k} \frac{1}{k^{2\nu_{ij}}} $. 
	Then, the conditions i) and iii) imply $ \nu_{ij} \leq \frac{1}{2} $. Especially, if $ \alpha_{ij,k} $ is polynomial, then $ \nu_{ij} < \frac{1}{2} $. 
\end{remark}

The following theorem proves the mean square convergence of \cref{algo:DE}. 

\begin{theorem}\label{thm:ms_conv}
	Under the condition of \cref{thm:as_conv}, the estimate $ \hat{\uptheta}_{i,k} $ in \cref{algo:DE} converges to the true value $ \theta $ in the mean square sense. 
\end{theorem}

\begin{proof}
	Since we have proved the almost sure convergence of \cref{algo:DE}, by Theorem 2.6.4 of \cite{Shiryaev}, it suffices to prove the uniform integrability of the algorithm. Here, we continue to use the notations of $ \mathtt{L}_{G,k} $, 
	$ \tilde{\Theta}_{k} $, $ \mathbb{H}_{\beta,k} $, and $ \mathtt{W}_k $ in the proof of \cref{thm:as_conv}. 
	

	\MODIFY{
	Denote $ \mathtt{A}_k = I_{N\times n} - \mathbb{H}_{\beta,k} - \mathtt{L}_{G,k}\otimes \varphi_k\varphi_k\tr $. When $ k $ is sufficiently large, $ \Abs{\mathtt{A}_k} \leq 1 $. Then, by \eqref{eq:recur_Theta1}, 
	\begin{align}\label{ineq:Aln1}
		& \mE \Absl{\tilde{\Theta}_k}^2 \ln \left( 1 + \Absl{\tilde{\Theta}_k}^2\right) \\
		\leq \ \ \ \ \ \ &\!\!\!\!\!\!\!  \mE \left( \Absl{\tilde{\Theta}_{k-1}}^2 + 2 \mathtt{W}_k\tr \mathtt{A}_k \tilde{\Theta}_{k-1} + \Absl{\mathtt{W}_k}^2 \right) \ln \left( 1 + \Absl{\tilde{\Theta}_{k-1}}^2 + 2 \mathtt{W}_k\tr \mathtt{A}_k \tilde{\Theta}_{k-1} + \Absl{\mathtt{W}_k}^2 \right). \nonumber 
	\end{align}
	By b) of \cref{lemma:ln} in \cref{appen}, 
	\begin{align}\label{ineq:Aln2}
		& \mE \Absl{\tilde{\Theta}_{k-1}}^2  \ln \left( 1 + \Absl{\tilde{\Theta}_{k-1}}^2 + 2 \mathtt{W}_k\tr \mathtt{A}_k \tilde{\Theta}_{k-1} + \Absl{\mathtt{W}_k}^2 \right) \\
		\leq &  \mE \Absl{\tilde{\Theta}_{k-1}}^2  \ln \left( 1 + \Absl{\tilde{\Theta}_{k-1}}^2 \right) + \mE \frac{\Absl{\tilde{\Theta}_{k-1}}^2}{1+\Absl{\tilde{\Theta}_{k-1}}^2}\left( 2 \mathtt{W}_k\tr \mathtt{A}_k \tilde{\Theta}_{k-1} + \Absl{\mathtt{W}_k}^2 \right) \nonumber\\
		\leq & \mE \Absl{\tilde{\Theta}_{k-1}}^2  \ln \left( 1 + \Absl{\tilde{\Theta}_{k-1}}^2 \right) + \mE \Absl{\tilde{\Theta}_{k-1}}^2 \mE \Absl{\mathtt{W}_k}^2. \nonumber
	\end{align}
	By a), c) and d) of \cref{lemma:ln} in \cref{appen}, 
	\begin{align} \label{ineq:Aln3}
		& \mE 2 \mathtt{W}_k\tr \mathtt{A}_k \tilde{\Theta}_{k-1} \ln \left( 1 + \Absl{\tilde{\Theta}_{k-1}}^2 + 2 \mathtt{W}_k\tr \mathtt{A}_k \tilde{\Theta}_{k-1} + \Absl{\mathtt{W}_k}^2 \right) \\
		\leq & \mE 2 \mathtt{W}_k\tr \mathtt{A}_k \tilde{\Theta}_{k-1} \ln \left( 1 + \Absl{\tilde{\Theta}_{k-1}}^2 + \Absl{\mathtt{W}_k}^2 \right) + \mE \left( 2 \mathtt{W}_k\tr \mathtt{A}_k \tilde{\Theta}_{k-1} \right)^2 \nonumber\\
		\leq & \mE 2 \mathtt{W}_k\tr \mathtt{A}_k \tilde{\Theta}_{k-1} \ln \left( 1 + \Absl{\tilde{\Theta}_{k-1}}^2 \right)  + 4 \mE \Absl{\mathtt{W}_k}^2 \mE \Absl{\tilde{\Theta}_{k-1}}^2 \nonumber\\
		& + \mE 2 \absl{\mathtt{W}_k\tr \mathtt{A}_k \tilde{\Theta}_{k-1}} \left( \ln \left( 1 + \Absl{\tilde{\Theta}_{k-1}}^2 + \Absl{\mathtt{W}_k}^2 \right) - \ln \left( 1 + \Absl{\tilde{\Theta}_{k-1}}^2 \right) \right) \nonumber\\
		\leq &  \mE 2 \Absl{\mathtt{W}_k} \Absl{\tilde{\Theta}_{k-1}} \ln \left( 1 + \Absl{\mathtt{W}_k}^2 \right) + 4 \mE \Absl{\tilde{\Theta}_{k-1}}^2 \mE \Absl{\mathtt{W}_k}^2  \nonumber\\
		\leq & O\left(\mE \Absl{\tilde{\Theta}_{k-1}} \mE \Absl{\mathtt{W}_k}^2   \right) + 4 \mE \Absl{\tilde{\Theta}_{k-1}}^2 \mE \Absl{\mathtt{W}_k}^2. \nonumber
	\end{align}
	By a) and d) of \cref{lemma:ln} in \cref{appen},  
	\begin{align}\label{ineq:Aln4}
		& \mE \Absl{\mathtt{W}_k}^2 \ln \left( 1 + \Absl{\tilde{\Theta}_{k-1}}^2 + 2 \mathtt{W}_k\tr \mathtt{A}_k \tilde{\Theta}_{k-1} + \Absl{\mathtt{W}_k}^2 \right) \\
		\leq & \mE \Absl{\mathtt{W}_k}^2 \ln \left( 1 + 2 \Absl{\tilde{\Theta}_{k-1}}^2 + 2 \Absl{\mathtt{W}_k}^2 \right) \nonumber\\
		\leq & \mE \Absl{\mathtt{W}_k}^2 \ln \left( 1 + 2 \Absl{\tilde{\Theta}_{k-1}}^2 \right) + \mE \Absl{\mathtt{W}_k}^2 \ln \left( 1 + 2 \Absl{\mathtt{W}_k}^2 \right) \nonumber\\
		\leq & 2 \mE\Absl{\tilde{\Theta}_{k-1}}^2  \mE \Absl{\mathtt{W}_k}^2 + O\left( \mE \Absl{\mathtt{W}_k}^{\min\{\rho,4\}} \right), \nonumber
	\end{align}
	where $ \rho $ is given in \cref{assum:noise}.}
%
	Taken the expectation over \eqref{eq:CondE_sumtheta}, we have
	$ \mE \Absl{\tilde{\Theta}_k}^2 $ is uniformly bounded. By Lyapunov inequality \cite{Shiryaev}, one can get $ \mE \Absl{\tilde{\Theta}_k} $ is also uniformly bounded. 
	\MODIFY{Besides, $ \mE \Absl{\mathtt{W}_k}^2 = O\left( \left( \sum_{i=1}^{N} \beta_{i,k}^2  + \sum_{(i,j)\in\mathcal{E}} \alpha_{ij,k}^2 \right) \right) $, and $ \mE \Absl{\mathtt{W}_k}^{\min\{\rho,4\}} = O\left( \left( \sum_{i=1}^{N} \beta_{i,k}^{\min\{\rho,4\}}  + \sum_{(i,j)\in\mathcal{E}} \alpha_{ij,k}^{\min\{\rho,4\}} \right) \right) $.
	Hence, \eqref{ineq:Aln1}-\eqref{ineq:Aln4} imply that $ \mE\Absl{\tilde{\Theta}_{k}}^2 \ln \left( 1 + \Absl{\tilde{\Theta}_{k}}^2 \right) $ is uniformly bounded. 	
	Note that
	\begin{align*}
		& \lim_{x\to\infty} \sup_{k\in\mathbb{N}} \int_{\{ \Absl{\tilde{\Theta}_{k}}^2 > x \}}  \Absl{\tilde{\Theta}_{k}}^2 \text{d} \mP \\
		\leq & \lim_{x\to\infty} \sup_{k\in\mathbb{N}} \frac{1}{\ln (1+x)} \int_{\{ \Absl{\tilde{\Theta}_{k}}^2 > x \}} \Absl{\tilde{\Theta}_{k}}^2 \ln \left( 1 + \Absl{\tilde{\Theta}_{k}}^2 \right) \text{d} \mP \\
		\leq & \lim_{x\to\infty}  \sup_{k\in\mathbb{N}} \frac{1}{\ln (1+x)} \mE \Absl{\tilde{\Theta}_{k}}^2 \ln \left( 1 + \Absl{\tilde{\Theta}_{k}}^2 \right)
		= 0. 
	\end{align*}}
	Then, $ \Absl{\tilde{\Theta}_{k}}^2 $ is uniformly integrable. Hence, the theorem can be proved by Theorem 2.6.4 of \cite{Shiryaev} and \cref{thm:as_conv}. 
%
\end{proof}


\begin{remark}
	If \eqref{condi:w_rho} holds for any $ \rho > 0 $, then similar to \cref{thm:ms_conv},  we can prove the $ L^r $ convergence of \cref{algo:DE} for any positive integer $ r $. 
\end{remark}

\begin{remark}
	Under finite data rate, existing literature \cite{sayin2014compressive,xie2013LMS} focuses on the mean square stability in terms of effectiveness, and gives the upper bounds of the mean square estimation errors for corresponding algorithms. There are two important breakthroughs in \cref{thm:as_conv,thm:ms_conv}. Firstly, \cref{thm:ms_conv} shows that our algorithm can not only achieve mean square stability, but also can achieve mean square convergence. The mean square estimation errors of our algorithm can converge to zero. Secondly, \cref{thm:as_conv} shows that the estimates of our algorithm can converge not only in the mean square sense, but also in the almost sure sense. The almost sure convergence property can better describe the characteristics of a single trajectory. When using our algorithm, there is no need to worry about the small probability event that the estimation errors do not converge to zero, as it will not occur almost surely.
\end{remark}

\subsection{Convergence rate}\label{ssec:conv_rate}

To quantitatively demonstrate the effectiveness, the following theorem calculates the almost sure convergence rate of \cref{algo:DE}. 

\begin{theorem}\label{thm:conv_rate}
	In \cref{algo:DE}, set $ \alpha_{ij,k} = \frac{\alpha_{ij,1}}{k^{\gamma_{ij}}} $ and $ \beta_{i,k} = \frac{\beta_{i,1}}{k} $ with
	\begin{enumerate}[label={\roman*)}]
		\item $ \alpha_{ij,1} = \alpha_{ji,1} > 0 $ for all $ (i,j)\in\mathcal{E} $, and $ \beta_{i,1} > 0 $ for all $ i \in \mathcal{V} $;
		\item $ 1/2 < \gamma_{ij} \leq 1 $ and $ \nu_{ij} + \gamma_{ij} \leq 1 $ for all $ (i,j) \in \mathcal{E} $. 
	\end{enumerate}
	Then, under \cref{assum:input,assum:noise,assum:connect,assum:dither}, the almost sure  convergence rate of the estimation error for the sensor $ i $ is
	\begin{align*}
		\tilde{\uptheta}_{i,k} = 
		\begin{cases}
			O\left( \frac{1}{k^{a}} \right), & \text{if}\ 2h-2a>1;\\
			O\left( \frac{\ln k}{k^{h-1/2}} \right), & \text{if}\ 2h-2a=1;\\
			O\left( \frac{\sqrt{\ln k}}{k^{h-1/2}} \right), & \text{if}\ 2h-2a<1,
		\end{cases}\ \as,
	\end{align*}
	where $ h = \min_{(i,j)\in\mathcal{E}} \left(\frac{\nu_{ij}}{2} + \gamma_{ij}\right) $, $ \lambda_2(\mathcal{L}) $ is defined in \eqref{eq:Laplacian_lower}, $ \mathcal{E}^\prime = \{ (i,j) \in \mathcal{E} : \nu_{ij} + \gamma_{ij} = 1 \} $, and
	\begin{align*}
		\displaystyle{a = \begin{cases}
			\frac{\delta \left(\min_{i\in\mathcal{V}} \beta_{i,1}\right)  }{N}, & \text{if}\ \mathcal{E}^\prime = \varnothing;\\
			\frac{ \delta \lambda_2(\mathcal{L}) \left(\min_{i\in\mathcal{V}} \beta_{i,1}\right) \left(\min_{(i,j)\in\mathcal{E}^\prime}\alpha_{ij,1} \frac{ \exp(-\Absl{\theta}_1/b_{ij})}{b_{ij}}\right) }{2Nn\bar{H}^2\left(\min_{i\in\mathcal{V}} \beta_{i,1}\right) + N \lambda_2(\mathcal{L}) \left(\min_{(i,j)\in\mathcal{E}^\prime} \alpha_{ij,1}\frac{ \exp(-\Absl{\theta}_1/b_{ij})}{b_{ij}}\right)}, & \text{if}\ \mathcal{E}^\prime \neq \varnothing.
		\end{cases}}
	\end{align*} 
\end{theorem}

\begin{proof}
	The key of the proof is to use \cref{lemma:iter_asrate} in \cref{appen}. 
	Here, we continue to use the notations of $ \mathtt{L}_{G,k} $, 
	$ \tilde{\Theta}_{k} $, $ \mathbb{H}_{k} $, $ \mathbb{H}_{\beta,k} $, $ \Phi_k $, and $ \mathtt{W}_k $ in the proof of \cref{thm:as_conv}. 
	Under the step-sizes in this theorem, by \eqref{eq:recur_Theta1}, one can get
	\begin{align}\label{eq:recur_Theta}
		\tilde{\Theta}_k = \left( I_{N\times n} - \frac{1}{k} \left( k \mathbb{H}_{\beta,k} +  k  \mathtt{L}_{G,k}\otimes \varphi_k\varphi_k\tr \right) \right) \tilde{\Theta}_{k-1} + \mathtt{W}_k.
	\end{align}

	Since $ \mE\left[ \left( \hat{\mathtt{s}}_{ji,k} - G_{ji,k}(\mathtt{x}_{j,k}) \right)^2 \middle| \mathcal{F}_{k-1} \right] = O\left( \frac{1}{k^{\nu_{ij}}} \right) $ almost surely, we have
	\begin{align*}
		\mE \left[ \Absl{\mathtt{W}}_k^2 \middle| \mathcal{F}_{k-1} \right] 
		= O\left(\frac{1}{k^2} + \frac{1}{k^{\min_{(i,j)\in\mathcal{E}}(\nu_{ij}+2\gamma_{ij})}} \right)
		= O\left( \frac{1}{k^{2h}} \right),\ \as
	\end{align*}
	Besides by \eqref{eq:def_udg}, one can get 
	\begin{align*}
		\mathtt{L}_{G,k} = O\left(\frac{1}{k^{\min_{(i,j)\in\mathcal{E}} (\nu_{ij}+\gamma_{ij})}}\right),\ \as
	\end{align*}
	Therefore, we have $ k \mathbb{H}_{\beta,k} + k  \mathtt{L}_{G,k}\otimes \varphi_k\varphi_k\tr = O\left( k^{1-\min_{(i,j)\in\mathcal{E}} (\nu_{ij}+\gamma_{ij})} \right) $ almost surely.  
	
	Firstly, we show that $ h \leq \min\left\{ 1, \frac{3+2h-2(1-\min_{(i,j)\in\mathcal{E}}(\nu_{ij}+\gamma_{ij}))}{3} \right\}$. Note that $ h = \min_{(i,j)\in\mathcal{E}} \left(\frac{\nu_{ij}}{2} + \gamma_{ij}\right) \leq \min_{(i,j)\in\mathcal{E}}(\nu_{ij}+\gamma_{ij}) $. Then, one can get $ h \leq 1 $ and 
	\begin{align*}
		h 
		< \frac{1+4h}{4}
		\leq \frac{ 3 + 2 h - 2\left(1 - \min_{(i,j)\in\mathcal{E}} (\nu_{ij} + \gamma_{ij})\right) }{4}.  
	\end{align*}

	Secondly, we estimate the lower bound of $ \frac{1}{np}\sum_{t=k-np+1}^{k}\left( t \mathbb{H}_{\beta,t} +  t \mathtt{L}_{G,t}\otimes \varphi_t\varphi_t\tr \right) $.
	By \eqref{eq:Laplacian_lower} and \eqref{ineq:underline_H}, one can get
	\begin{align*}
		\sum_{t=k-np+1}^{k}\left( t \mathbb{H}_{\beta,t} +  t \mathtt{L}_{G,t}\otimes \varphi_t\varphi_t\tr \right)
		\geq z_1 \sum_{t=k-np+1}^{k} \Phi_t
		\geq \underline{\mathtt{H}} > 0,\ \as,
	\end{align*}
	where $ z_1 = \min\left\{ \alpha_{ij,1}, (i,j)\in\mathcal{E}; \beta_{i,1}, i \in\mathcal{V} \right\} $. Then, by  \cref{lemma:iter_asrate}, $ \tilde{\Theta}_{k} = O\left( \frac{1}{k^\psi} \right) $ for some $ \psi>0 $ almost surely. Hence, by the Lagrange mean value theorem \cite{zorich} and \cref{lemma:inf_kf}, we have $ g_{ij}(\upxi_{ij,k}) - g_{ij}(\varphi_{k}\tr \theta) = O\left( \frac{1}{k^{\nu_{ij}+\psi}} \right) $ almost surely, which implies 
	\begin{equation}\label{ineq:g}
		g_{ij} (\upxi_{ij,k}) 
		\!\geq\! \frac{\exp\left( \frac{-\absl{\varphi_k\tr\theta} -C_{ij,k}}{b_{ij}}\right)}{b_{ij}} + O\left( \frac{1}{k^{\nu_{ij}+\psi}} \right)
		\!\geq\! \frac{e^{-\Absl{\theta}_1/b_{ij}}}{b_{ij}k^{\nu_{ij}}} + O\left( \frac{1}{k^{\nu_{ij}+\psi}} \right),\as 
	\end{equation}
	
	By \cref{assum:input}, \eqref{ineq:g}, and Lemma 5.4 of \cite{xie2018analysis}, it holds that
	\begin{align}\label{eq:sum_P1}
		& \sum_{t=k-np+1}^{k}\left( t \mathbb{H}_{\beta,t} +  t \mathtt{L}_{G,t}\otimes \varphi_t\varphi_t\tr \right) \\
		\geq & \sum_{t=k-np+1}^{k} \left( t \mathbb{H}_{\beta,t} +  R_t \left( I_N - J_N \right) \otimes \varphi_t\varphi_t\tr \right) \nonumber\\
		\geq &  \sum_{t=k-np+1}^{k} \left( \min_{i\in\mathcal{V}} \beta_{i,1}\right) \mathbb{H}_{t} + \left(\min_{k-np+1\leq t \leq k} R_t\right) \left( I_N - J_N \right) \otimes \left( \sum_{t=k-np+1}^{k} \varphi_t\varphi_t\tr\right) \nonumber\\
		= &  \sum_{t=k-np+1}^{k}\left( \left( \min_{i\in\mathcal{V}} \beta_{i,1}\right) \mathbb{H}_{t} +  \frac{1}{n} \left(\min_{k-np+1\leq t \leq k} R_t\right) \left( I_N - J_N \right) \otimes I_n \right) \nonumber\\
		\geq& \frac{np\delta \left(\min_{i\in\mathcal{V}} \beta_{i,1}\right) \left(\min_{k-np+1\leq t \leq k} R_t\right)}{2Nn\bar{H}^2\left(\min_{i\in\mathcal{V}} \beta_{i,1}\right)+ N\left(\min_{k-np+1\leq t \leq k} R_t\right)} I_{Nn},  \nonumber\\
		=& npa I_{Nn} + O\left( \frac{1}{\psi^\prime}\right), \as, \nonumber
	\end{align}
	for some $ \psi^\prime > 0 $, 
	where $ R_k = \left( \min_{(i,j)\in\mathcal{E}} \alpha_{ij,1} \frac{e^{-\Absl{\theta}_1/b_{ij}}}{b_{ij}} k^{1-\nu_{ij}-\gamma_{ij}} \left( 1+O\left(\frac{1}{k^\psi}\right) \right) \right)  \lambda_2(\mathcal{L}) $ and $ J_N $ is defined in \eqref{eq:Laplacian_lower}. 

	Then, by \eqref{eq:recur_Theta} and \cref{lemma:iter_asrate}, we have
	\begin{align*}
		\tilde{\Theta}_k = 
		\begin{cases}
			O\left( \frac{1}{k^{a}} \right), & \text{if}\ 2h-2a>1;\\
			O\left( \frac{\ln k}{k^{h-1/2}} \right), & \text{if}\ 2h-2a=1;\\
			O\left( \frac{\sqrt{\ln k}}{k^{h-1/2}} \right), & \text{if}\ 2h-2a<1,
		\end{cases}\ \as
	\end{align*}
\end{proof}


\begin{remark}
	\MODIFY{Given $ \nu_{ij} $ and $  \gamma_{ij} $, an almost sure convergence rate of $ O\left( \frac{\sqrt{\ln k}}{k^{h-1/2}} \right) $ can be achieved by properly selecting $ \alpha_{ij,1} $, $ \beta_{i,1} $ and $ b_{ij} $. Especially, }
	when $ \nu_{ij} = 0 $, $ \gamma_{ij} = 1 $, and $ a $ is sufficiently large, \cref{algo:DE} can achieve an almost sure convergence rate of $ O( \sqrt{\ln k / k} ) $, which is the best one among existing literature \cite{he2022event,kar2013distributed,ZhangQ2012} even without data rate constraints. 
	For comparison, He et al. \cite{he2022event} and Kar et al. \cite{kar2013distributed} show that their distributed estimation algorithm achieve a almost sure convergence rate of $ o\left( k^{-\tau} \right) $ for some $ \tau \in [0,\frac{1}{2}) $. Zhang and Zhang \cite{ZhangQ2012} prove that $ \frac{1}{k} \sum_{t=1}^{k} \Absl{\tilde{\Theta}_t} = o\left( (b(k)k)^{-1/2} \right) $ almost surely for their algorithm, where $ b(k) $ is the step-size satisfying the stochastic approximation condition $ \sum_{k=0}^{\infty} b(k) = \infty, \sum_{k=0}^{\infty} b^2(k) < \infty $. 
	The theoretical result of \cref{thm:conv_rate} is better than these ones. 
	\MODIFY{
		Our technique can be applied in the almost sure convergence rate analysis of other distributed estimation algorithms. 
		For example, if the step-size $ b(t) $ in the distributed estimation algorithm (3) of \cite{ZhangQ2012} is selected as $ \frac{\beta}{k} $ with sufficiently large $ \beta $, then by \cref{lemma:iter_asrate}, an almost sure convergence rate of $ O( \sqrt{\ln k / k} ) $ can also be achieved.}
\end{remark}

\begin{remark}
	When $ \nu_{ij} < 1 $ for some $ (i,j) \in \mathcal{E} $, we have $ h = \min_{(i,j)\in\mathcal{E}} \left(\frac{\nu_{ij}}{2} + \gamma_{ij}\right)$  $ < 1 $. Therefore, the almost sure convergence rate of $ O( \sqrt{\ln k / k} ) $ cannot be obtained. This is because the communication frequency is reduced. Similar results can be seen in \cite{he2022event}. The trade-off between the convergence rate and the communication cost is discussed in \cref{sec:trade-off}. 
\end{remark}

\section{Communication cost}\label{sec:comm_cost}

This section analyzes the communication cost of \cref{algo:DE} by calculating the average data rates defined in \cref{def:bit}. 

Firstly, the local average data rates of \cref{algo:DE} are calculated. 

\begin{theorem}\label{thm:local_bit_rate}
	Under the condition of \cref{thm:as_conv}, the local average data rate $ \mathtt{B}_{ij}(k) = O\left(\frac{1}{k^{\nu_{ij}}}\right) $ almost surely. Furthermore, if $ \nu_{ij} = 0 $, then $ \mathtt{B}_{ij}(k) = 1 $. And, if $ \nu_{ij} > 0 $ and the step-sizes are set as \cref{thm:conv_rate} and $ a > h - 1/2 $, then 
	\begin{align*}
		\mathtt{B}_{ij}(k) \leq \frac{\exp\left( \Abs{\theta}_1 / b_{ij} \right)}{(1-\nu_{ij}) k^{\nu_{ij}}} + O\left( \frac{\sqrt{\ln k}}{k^{h-1/2+\nu_{ij}}} \right), \ \as 
	\end{align*}
\end{theorem}

\begin{proof}
	If $ \nu_{ij} = 0 $, then $ C_{ij,k} = 0 $. In this case, the sensor $ i $ transmits 1 bit of message to the sensor $ j $ at every moment almost surely, which implies $ \mathtt{B}_{ij}(k) = 1 $ almost surely. Therefore, it suffices to discuss the case of $ \nu_{ij} > 0 $. 
	
	
	By the definition of $ \upzeta_{ij}(k) $, we have $ \upzeta_{ij}(k) $ is $ \mathcal{F}_k $-measurable, and
	\begin{align*}
		\mP\{\upzeta_{ij}(k) = 1 \} =& F\left(\frac{\mathtt{x}_{i,k}-C_{ij,k}}{b_{ij}}\right) + F\left(\frac{-\mathtt{x}_{i,k}-C_{ij,k}}{b_{ij}}\right), \\
		\mP\{\upzeta_{ij}(k) = 0 \} =& 1 - F\left(\frac{\mathtt{x}_{i,k}-C_{ij,k}}{b_{ij}}\right) - F\left(\frac{-\mathtt{x}_{i,k}-C_{ij,k}}{b_{ij}}\right). 
	\end{align*}
	
	Firstly, we estimate $ \sum_{t=1}^{k} \mE \left[\upzeta_{ij}(t) \middle| \mathcal{F}_{t-1} \right] $. 
	By \cref{thm:as_conv}, $ \mathtt{x}_{i,k} = \varphi_k\tr \hat{\uptheta}_{i,k} $ is uniformly bounded almost surely.  
	Therefore, 
	when $ k $ is sufficiently large,
	\begin{align}\label{eq:condi_kappa}
		\mE \left[\upzeta_{ij}(k) \middle| \mathcal{F}_{k-1} \right] 
		= & F\left(\frac{\mathtt{x}_{i,k}-C_{ij,k}}{b_{ij}}\right) + F\left(\frac{-\mathtt{x}_{i,k}-C_{ij,k}}{b_{ij}}\right)	\\
		= & \frac{\exp\left((\mathtt{x}_{i,k}-C_{ij,k})/b_{ij}\right)+\exp\left((-\mathtt{x}_{i,k}-C_{ij,k})/b_{ij}\right)}{2} \nonumber \\
		= & \frac{\exp\left(\mathtt{x}_{i,k}/b_{ij}\right)+\exp\left(-\mathtt{x}_{i,k}/b_{ij}\right)}{2k^{\nu_{ij}}}
		= O\left( \frac{1}{k^{\nu_{ij}}} \right),\ \as \nonumber
	\end{align}
	Hence, $ \mE \left[\upzeta_{ij}(k) \middle| \mathcal{F}_{k-1} \right] = O\left( \frac{1}{k^{\nu_{ij}}} \right) $ for $ \nu_{ij} \geq 0 $ almost surely, which implies
	\begin{align}\label{eq:sum_condi_kappa}
		\sum_{t=1}^{k} \mE \left[\upzeta_{ij}(t) \middle| \mathcal{F}_{t-1} \right] = O\left( k^{1-\nu_{ij}} \right),\ \as 
	\end{align}
	
	Secondly, we estimate $ \sum_{t=1}^{k} \upzeta_{ij}(t)-\mE \left[\upzeta_{ij}(t) \middle| \mathcal{F}_{t-1} \right] $. Since $ \nu_{ij} \leq \frac{1}{2} $ under the condition of \cref{thm:as_conv},  $ 1 - \nu_{ij} > \frac{1}{2} - \frac{\nu_{ij}}{4} $. 
	By $ \mE \left[\upzeta_{ij}(k) \middle| \mathcal{F}_{k-1} \right] = O\left( \frac{1}{k^{\nu_{ij}}} \right) $ almost surely and $ \upzeta_{ij}(k) = 0 $ or $ 1 $, we have
	\begin{align*}
		& \mE \left[ \Absl{\upzeta_{ij}(k)-\mE \left[\upzeta_{ij}(k) \middle| \mathcal{F}_{k-1} \right]}^{4} \middle| \mathcal{F}_{k-1}  \right] \\ 
		\leq & \mE \left[ \left( \upzeta_{ij}(k)-\mE \left[\upzeta_{ij}(k) \middle| \mathcal{F}_{k-1} \right] \right)^2 \middle| \mathcal{F}_{k-1}  \right] \\
		= & \mE \left[ \upzeta_{ij} \middle| \mathcal{F}_{k-1}  \right] - \left( \mE \left[ \upzeta_{ij}(k) \middle| \mathcal{F}_{k-1}  \right] \right)^2 
		= O\left( \frac{1}{k^{\nu_{ij}}} \right),\ \as
	\end{align*}
	Then, by Theorem 1.3.10 of \cite{GuoL2020}, it holds that
	\begin{align}\label{eq:sum_kappa-kappa}
		&\sum_{t=1}^{k} \left(\upzeta_{ij}(t)-\mE \left[\upzeta_{ij}(t) \middle| \mathcal{F}_{t-1} \right]\right) \\
		= & \sum_{t=1}^{k} \frac{1}{t^{\nu_{ij}/4}} \cdot t^{\nu_{ij}/4} \left(\upzeta_{ij}(t)-\mE \left[\upzeta_{ij}(t) \middle| \mathcal{F}_{t-1} \right]\right)
		= O\left( k^{\frac{1}{2}-\frac{\nu_{ij}}{4}}\sqrt{\ln \ln k} \right),\ \as \nonumber
	\end{align}
	
%

	\eqref{eq:sum_condi_kappa} and \eqref{eq:sum_kappa-kappa} imply $ \sum_{t=1}^{k} \upzeta_{ij}(t) = O(k^{1-\nu_{ij}}) $ almost surely. 
	Therefore, $ \mathtt{B}_{ij}(k) = O\left(\frac{1}{k^{\nu_{ij}}}\right) $ almost surely.
	
	If the step-sizes are set as \cref{thm:conv_rate} and $ a > h - 1/2 $, then by \cref{thm:conv_rate}, $ \tilde{\uptheta}_{i,k} = O \left( \frac{\sqrt{\ln k}}{k^{h-1/2}} \right) $ almost surely for all $ i \in \mathcal{V} $. Then, by \eqref{eq:condi_kappa}, we have
	\begin{align*}
		\mE \left[ \upzeta_{ij}(k) \middle| \mathcal{F}_{k-1} \right] \leq \frac{\exp\left( \Abs{\theta}_1 / b_{ij} \right)}{k^{\nu_{ij}}} + O\left( \frac{\sqrt{\ln k}}{k^{h-1/2+\nu_{ij}}} \right), \ \as 
	\end{align*}
	Therefore, one can get
	\begin{align*}
		\mathtt{B}_{ij}(k) \leq \frac{\exp\left( \Abs{\theta}_1 / b_{ij} \right)}{(1-\nu_{ij}) k^{\nu_{ij}}} + O\left( \frac{\sqrt{\ln k}}{k^{h-1/2+\nu_{ij}}} \right),\ \as 
	\end{align*}
\end{proof}

\begin{remark}
	By \cref{thm:local_bit_rate}, the decaying rate of $ \mathtt{B}_{ij}(k) $ only depends on $ \nu_{ij} $.  Therefore, the operators of sensors $ i $ and $ j $ can directly set and easily know the decaying rate of $ \mathtt{B}_{ij}(k) $ before running the algorithm.
\end{remark}

\MODIFY{\begin{remark}
		The noise coefficient $ b_{ij} $ influences the almost sure convergence rate and the average data rate. By \cref{thm:conv_rate}, an almost sure convergence rate of $ O\left( \frac{\sqrt{\ln k}}{k^{h-1/2}} \right) $ can be achieved when $ 2h-2a < 1 $, where $ a $ is a function of $ b_{ij} $. By \cref{thm:local_bit_rate}, the upper bound of $ \mathtt{B}_{ij}(k) $ is monotonically non-increasing with $ b_{ij} $. Therefore, increasing $ b_{ij} $ while maintaining $ 2h - 2a < 1 $ can reduce the communication cost without losing the almost sure convergence rate.
\end{remark}}

Then, we can estimate the global average data rate. 

\begin{theorem}\label{thm:global_bit}
	Under the condition of \cref{thm:as_conv}, the global average data rate
	$ \mathtt{B}(k) = O\left(\frac{1}{k^{\underline{\nu}}}\right) $ almost surely, 
	where $ \underline{\nu} = \min_{(i,j)\in\mathcal{E}} \nu_{ij} $. 
\end{theorem}

\begin{proof}
	The theorem can be proved by \cref{thm:local_bit_rate} and $ \mathtt{B}(k) = \frac{\sum_{(i,j)\in\mathcal{E}} \mathtt{B}_{ij}(k)}{2M} $.	
\end{proof}

\begin{remark}
	If the step-sizes are set as \cref{thm:conv_rate} and $ a > h - 1/2 $, the upper bound of global average data rate $ \mathtt{B}(k) $ can also be obtained by \cref{thm:local_bit_rate} and $ \mathtt{B}(k) = \frac{\sum_{(i,j)\in\mathcal{E}} \mathtt{B}_{ij}(k)}{2M} $.  
\end{remark}

\section{Trade-off between convergence rate and communication cost}\label{sec:trade-off}

In \cref{sec:conv,sec:comm_cost}, we quantitatively demonstrate the effectiveness of \cref{algo:DE} by the almost sure convergence rate and the communication cost by the average data rates. This section establishes the trade-off between the convergence rate and the communication cost.

By \cref{thm:conv_rate}, the convergence rate of \cref{algo:DE} is influenced by the selection of step-sizes $ \alpha_{ij,k} $ and $ \beta_{i,k} $. The following theorem optimizes almost sure convergence rate by properly selecting the step-sizes. 

\begin{theorem}\label{thm:coro_conv_rate}
	In \cref{algo:DE}, set $ \nu_{ij} \in [0,\frac{1}{2}) $. Then, under the condition of \cref{thm:as_conv}, there exist step-sizes $ \alpha_{ij,k} $ and $ \beta_{i,k} $ such that $ \tilde{\uptheta}_{i,k} = O\left(\frac{\sqrt{\ln k}}{k^{1/2- \bar{\nu}/2} }\right) $ almost surely, where $ \bar{\nu} = \max_{(i,j)\in\mathcal{E}} \nu_{ij} $. 
\end{theorem}

\begin{proof}
	Set $ \gamma_{ij} = 1 - \nu_{ij} $. Then, $ h $ in \cref{thm:conv_rate} equals to $ 1 - \bar{\nu}/2 $. Besides, when $ \alpha_{ij,1} $ and $ \beta_{i,1} $ are sufficiently large, $ a $ in \cref{thm:conv_rate} is larger than $ 2h - 1 $. Then, the theorem can be proved by \cref{thm:conv_rate}. 
\end{proof}

\begin{remark}
	The proof of \cref{thm:coro_conv_rate} provides a selection method to optimize the convergence rate of the algorithm. 
\end{remark}

\cref{thm:coro_conv_rate} shows that when properly selecting the step-sizes, the key factor to determine the almost sure convergence rate of \cref{algo:DE} is the event-triggered coefficient $ \nu_{ij} $. The optimal almost sure convergence rate of \cref{algo:DE} gets faster under smaller $ \nu_{ij} $.

On the other hand, \cref{thm:local_bit_rate} shows that $ \nu_{ij} $ is the decaying rate of the local average data rate for the communication channel $ (i,j)\in\mathcal{E} $. \cref{thm:global_bit} shows that  $ \underline{\nu} = \min_{(i,j)\in\mathcal{E}} \nu_{ij} $ is the decaying rate of the global average data rate. 
Therefore, the average data rates of \cref{algo:DE} get smaller under large $ \nu_{ij} $. 

Therefore, there is a trade-off between the convergence rate and the communication cost. The operator of each sensor $ i $ can decrease $ \nu_{ij} $ of the adjacent communication channel $ (i,j)\in\mathcal{E} $ for a better convergence rate, or increase $ \nu_{ij} $ for a lower communication cost.

\section{Simulation}\label{sec:simu}

This section gives a numerical example to illustrate the effectiveness and the average data rates of \cref{algo:DE}. 

Consider a network with $ 8 $ sensors. The communication topology is shown in \cref{fig:graph}. $ a_{ij} = 1 $ if $ (i,j) \in \mathcal{E} $, and $ 0 $, otherwise. For the sensor $ i $, the measurement matrix $ H_{i,k} = 
\begin{bmatrix}
	1 & 0
\end{bmatrix} $ if $ i $ is odd, and  $ 
\begin{bmatrix}
	0 & 1
\end{bmatrix} $ if $ i $ is even. The observation noise $ w_{i,k} $ is i.i.d. Gaussian with zero mean and standard deviation $ 0.1 $. The true value $ \theta = \begin{bmatrix}
	1 & -1
\end{bmatrix}^\top $. 

\begin{figure}[!htbp]
	\centering
	\begin{tikzpicture}[node distance = 0.5cm]
		\node[terminal](Agent 1) at (0,0) {1};
		\node[terminal, right=1.5cm of Agent 1](Agent 2){2};
		\node[terminal, right=1.5cm of Agent 2](Agent 3){3};
		\node[terminal, right=1.5cm of Agent 3](Agent 4){4};
		\node[terminal, below=1.5cm of Agent 4](Agent 5){5};
		\node[terminal, left=1.5cm of Agent 5](Agent 6){6};
		\node[terminal, left=1.5cm of Agent 6](Agent 7){7};
		\node[terminal, left=1.5cm of Agent 7](Agent 8){8};
		
		\draw[thick](Agent 1)--(Agent 2);
		\draw[thick](0.226,-0.226)--(1.914,-1.914);
		\draw[thick](6.1915,-0.226)--(4.5045,-1.914);
		\draw[thick](Agent 2)--(Agent 3);
		\draw[thick](Agent 3)--(Agent 4);
		\draw[thick](Agent 4)--(Agent 5);
		\draw[thick](Agent 5)--(Agent 6);
		\draw[thick](Agent 6)--(Agent 7);
		\draw[thick](Agent 7)--(Agent 8);
		\draw[thick](Agent 8)--(Agent 1);
		\draw[thick](Agent 2)--(Agent 7);
		\draw[thick](Agent 3)--(Agent 6);
	\end{tikzpicture}
	\caption{Communication topology.}\label{fig:graph}
\end{figure}

In \cref{algo:DE}, set $ b_{ij} = \frac{1}{2} $ and $ \nu_{ij} = \frac{1}{4} $. The step-sizes $ \alpha_{ij,k} = \frac{5}{k^{3/4}} $ and $ \beta_{i,k} = \frac{5}{k} $. \cref{fig:sumerror} shows the trajectory of $ \frac{1}{N} \sum_{i=1}^{N} \Absl{\tilde{\uptheta}_{i,k}}^2 $, which demonstrates the convergence of \cref{algo:DE}. 

\begin{figure}[!htbp]
	\centering
	\includegraphics[width=0.6\linewidth]{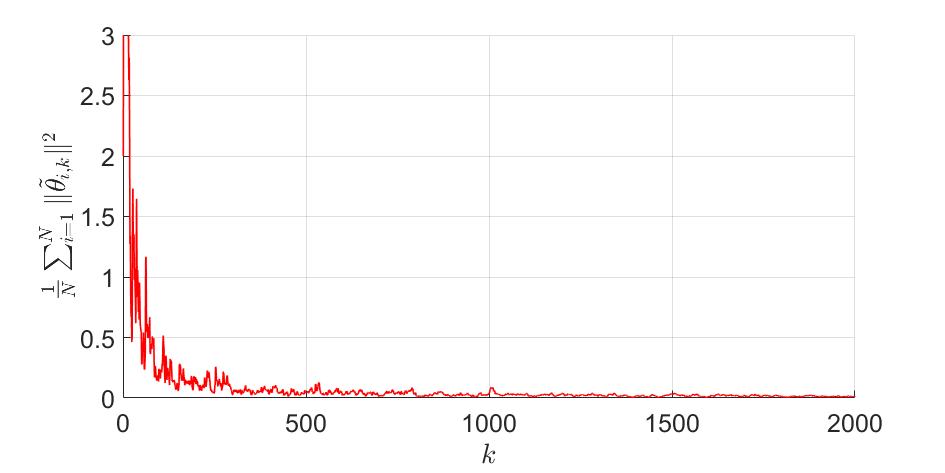}
	\caption{\MODIFY{The trajectory of $ \frac{1}{N} \sum_{i=1}^{N} \Absl{\tilde{\uptheta}_{i,k}}^2 $}}
	\label{fig:sumerror}
\end{figure}

To show the balance between the convergence rate and the communication cost, set $ b_{ij} = \frac{1}{2} $, \MODIFY{$ \nu_{ij} = \nu = 0,\frac{1}{9},\frac{2}{9},\frac{3}{9},\frac{4}{9} $}, and the step-sizes $ \alpha_{ij,k} = \frac{5}{k^{1-\nu}} $ and $ \beta_{i,k} = \frac{5}{k} $. \MODIFY{The simulation is repeated 50 times. Denote $ \tilde{\uptheta}_{i,k}^t $ as the estimation error of the sensor $ i $ at time $ k $ in the $ t $-th run. 
\cref{fig:conv_rate} depicts the log-log plot of $ \frac{1}{N} \sum_{i=1}^{N} \Absl{\tilde{\uptheta}_{i,k}^t}^2 $, which demonstrates that the convergence rate is faster under a smaller $ \nu $.  \cref{fig:bit_rate} shows the log-log plot of $ \mathtt{B}(k) $, which illustrates that the global average data rate is smaller under a larger $ \nu $. \cref{fig:conv_rate,fig:bit_rate} reveal the trade-off between the convergence rate and the data rate.}


\begin{figure}[!htbp]
	\centering
	\begin{minipage}[c]{0.5\textwidth}
		\centering
		\includegraphics[width=0.95\linewidth]{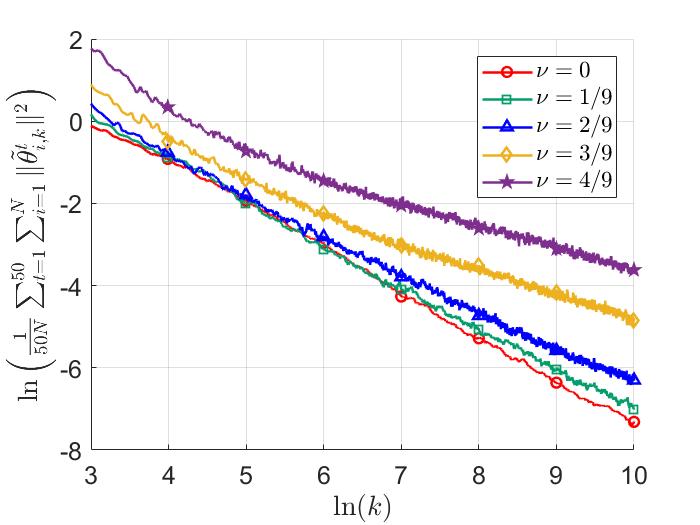}
		\caption{\MODIFY{Convergence rates with different $ \nu $}}
		\label{fig:conv_rate}
	\end{minipage}%
	\begin{minipage}[c]{0.5\textwidth}
		\centering
		\includegraphics[width=0.95\linewidth]{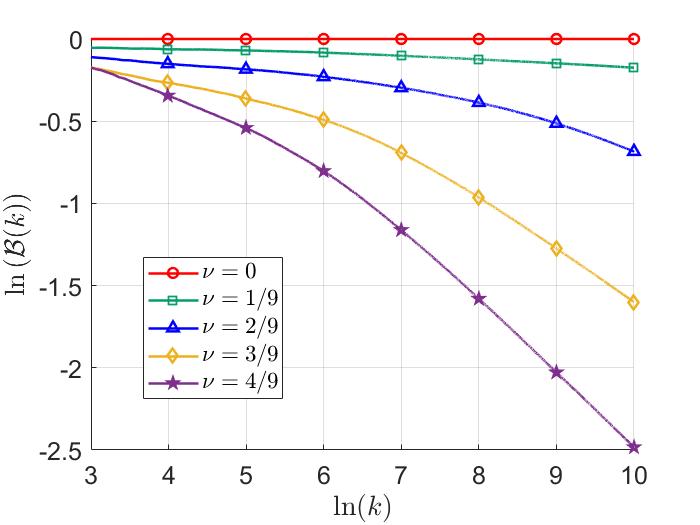}
		\caption{\MODIFY{Average data rates with different $ \nu $}}
		\label{fig:bit_rate}
	\end{minipage}
\end{figure}

\MODIFY{\cref{fig:conv_rate_c,fig:bit_rate_c} compare \cref{algo:DE} with the single bit diffusion algorithm \cite{sayin2014compressive} and the distributed least mean square (LMS) algorithm \cite{xie2013LMS}, which demonstrates that  \cref{algo:DE} can achieve higher estimation accuracy at a lower communication data rate compared to the algorithms in \cite{sayin2014compressive,xie2013LMS}. }

\begin{figure}[!htbp]
	\centering
	\begin{minipage}[c]{1\textwidth}
		\centering
		\includegraphics[width=0.7\linewidth]{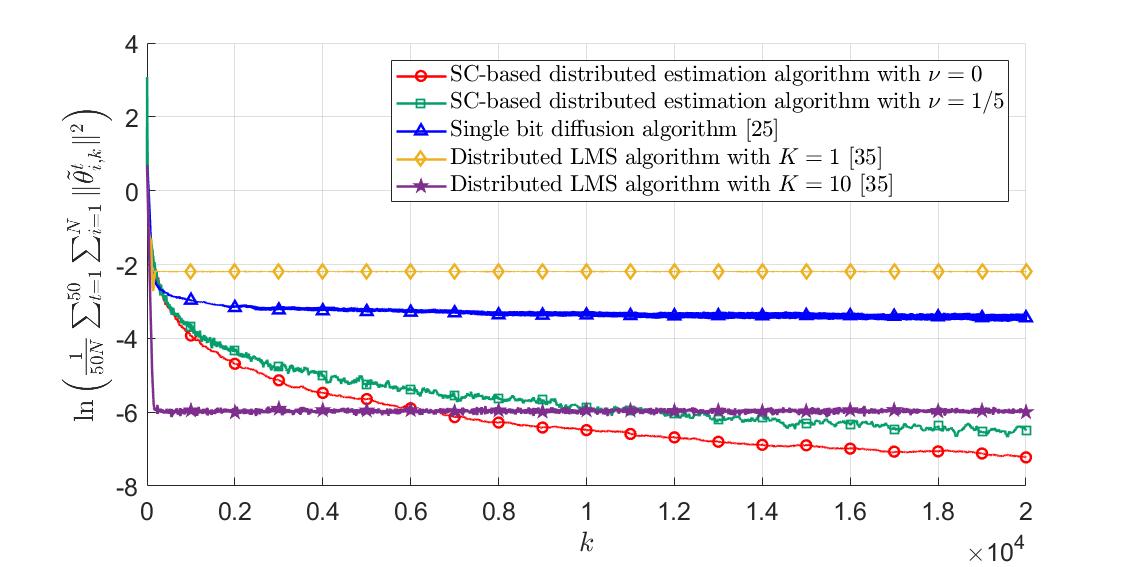}
		\caption{\MODIFY{The trajectories of $ \ln \left( \frac{1}{50N} \sum_{t=1}^{50} \sum_{i=1}^{N} \Absl{\tilde{\uptheta}_{i,k}^t}^2 \right) $ or different algorithms}}
		\label{fig:conv_rate_c}
	\end{minipage}%

	\begin{minipage}[c]{1\textwidth}
		\centering
		\includegraphics[width=0.7\linewidth]{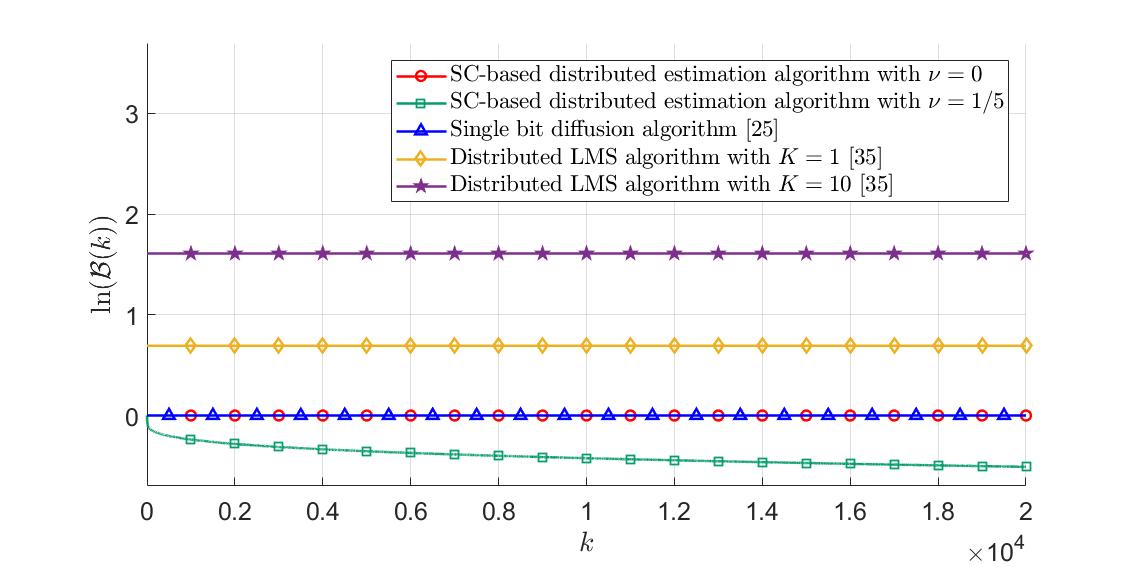}
		\caption{\MODIFY{Average data rates for different algorithms}}
		\label{fig:bit_rate_c}
	\end{minipage}
\end{figure}

\section{Conclusion}\label{sec:concl}

This paper considers the distributed estimation under low communication cost, which is described by the average data rates. We propose a novel distributed estimation algorithm, where the SC consensus protocol \cite{Ke2023Signal} is used to fuse neighborhood information, and a new stochastic event-triggered mechanism is designed to reduce the communication frequency. The algorithm has advantages both in the effectiveness and communication cost. 
For the effectiveness, the estimates of the algorithm are proved to converge to the true value in the almost sure and mean square sense, and polynomial almost sure convergence rate is also obtained. For the communication cost, the local and global average data rates are proved to decay to zero at polynomial rates. Besides, the trade-off between convergence rate and communication cost is established through event-triggered coefficients.  A better convergence rate can be achieved by decreasing event-triggered coefficients, while lower communication cost can be achieved by increasing event-triggered coefficients.  

There are interesting issues for future works. For example, how to extend the results to the cases with more complex communication graphs, such as directed graphs and switching graphs? Besides, Gan and Liu \cite{gan2022distributed} consider the distributed order estimation, and Xie and Guo \cite{xie2018analysis} investigate distributed adaptive filtering. These issues also suffer the communication cost problems. Then, how to apply our technique to these works to save the communication cost?


\appendix

\section{Lemmas}\label[appen]{appen}

%

\begin{lemma}\label{lemma:inf_kf}
	Let $ f(\cdot) $ be the density function of $ Lap(0,1) $.
	Given $ C_k = \nu b \ln k $ with $ \nu \geq 0 $ and $ b > 0 $, and a compact set $ \mathcal{X} $, we have $ \inf_{x\in\mathcal{X}, k\in\mathbb{N}} \frac{k^\nu}{b} f((x-C_k)/b) > 0 $. 
\end{lemma}


\begin{proof}
	If $ \nu = 0 $, then $ C_k = 0 $ for all $ k $. Therefore, $ \inf_{x\in\mathcal{X}, k\in\mathbb{N}} \frac{1}{b} f(x/b) > 0 $ by the compactness of $ \mathcal{X} $. 
	
	If $ \nu > 0 $, then $ \lim_{k\to\infty} C_k = \infty $, which together with the compactness of  $ \mathcal{X} $ implies that  there exists $ k_0 $ such that $ x - C_k < 0 $ for all $ x \in \mathcal{X} $ and $ k \geq k_0 $. Hence, 
	\begin{align*}
		\inf_{x\in\mathcal{X},k\geq k_0} \frac{k^\nu}{b} f\left(\frac{x-C_k}{b}\right)
		= \inf_{x\in\mathcal{X},k\geq k_0}  \frac{k^\nu}{2b}e^{(x-\nu b \ln k)/b}
		= \frac{1}{2b} e^{\min \mathcal{X}/b} > 0. 
	\end{align*} 
	Besides by the compactness of $ \mathcal{X} $, one can get $ \inf_{x\in\mathcal{X}} \frac{k^\nu}{b} f((x-C_k)/b) > 0 $ for all $ k < k_0 $. The lemma is proved. 
\end{proof}

%

\vspace{-2pt}

\begin{lemma}\label{lemma:sum_min_z}
	If positive sequence $ \{z_k\} $ satisfies $ \sum_{k=1}^{\infty} z_k = \infty $ and $ z_{k+1} = O(z_{k}) $, then for any $ l \in \{1,\ldots,n\} $, $ \sum_{q=1}^{\infty} \min_{n(q-1)+l<t\leq nq+l} z_{t} = \infty $. 
\end{lemma}

\begin{proof}
	Set $ \bar{z} = \sup\left\{ 1, \frac{z_k+1}{z_{k}}, k\in\mathbb{N} \right\} < \infty $. Then, $ z_{k} \geq \frac{z_{k+1}}{\bar{z}} $. Therefore, 
	\begin{align*}
		\sum_{q=1}^{\infty} \min_{n(q-1)+l<t\leq nq+l} z_{t}
		\geq \sum_{q=1}^{\infty} \max_{nq+l<t\leq n(q+1)+l} \frac{z_{t}}{\bar{z}^{2n}} 
		\geq \frac{1}{n\bar{z}^{2n}} \sum_{k=l+n+1}^{\infty} z_k = \infty. 
	\end{align*}
\end{proof}

\vspace{-2pt}
\MODIFY{
	\begin{lemma}\label{lemma:ln}
		a) $ \ln (1+x+y) \leq \ln (1+x) + \ln (1+y) $ for all $ x,y \geq 0 $; 
		
		b) $ \ln (1+x)-\ln (1+y) \leq \frac{x-y}{1+y} $ for all $ x,y \geq 0 $;
		 
		c) $ \frac{\ln (1+x)-\ln (1+y)}{x-y} \leq 1 $ for all $ x,y \geq 0 $;
		
		d) $ \sup_{x > 0} \frac{\ln (1+x)}{x^p}  < \infty $ for all $ p \in (0,1] $.
	\end{lemma} 

	\begin{proof}
		a), b) and c) can be proved by $ \ln (1+x+y) \leq \ln\left( (1+x) (1+y) \right) = \ln (1+x) + \ln (1+y) $, Proposition 5.4.6 of \cite{zorich} and the Lagrange mean value theorem \cite{zorich}, respectively. 
		For d), if $ p = 1 $, then we have $ \sup_{x \geq 0} \frac{\ln (1+x)}{x} \leq 1 $. If $ p \in (0,1) $, then $ x_1 > x_2 > 0 $ such that $ \ln (1+x) < x^p $ for all $ x \in (0,x_2) \cup (x_1,\infty) $. Therefore, $ \sup_{x \geq 0} \frac{\ln (1+x)}{x^p} \leq \max\left\{\sup_{x \in [x_2,x_1 ]} \frac{\ln (1+x)}{x^p}, 1\right\} < \infty $. 
	\end{proof} 
}

\begin{lemma}\label{lemma:iter_asrate}
	Assume that
	\begin{enumerate}[label={\roman*)}]
		\item $ \{\mathcal{F}_k\} $ is a $ \sigma $-algebra sequence satisfying $ \mathcal{F}_{k-1} \subseteq \mathcal{F}_k $ for all $ k $;	
		\MODIFY{\item $ \{\mathtt{U}_k\} $ is a matrix sequence satisfying that $ \mathtt{U}_k $ is $ \mathcal{F}_{k-1} $-measurable, $ \mathtt{U}_k = O\left( k^{\mu} \right) $ for  some $ 0 \leq \mu < \frac{1}{2} $ almost surely, $ \mathtt{U}_k + \mathtt{U}_k\tr $ is positive semi-definite for all $ k $, and
			\vspace{-2pt}
		\begin{equation}\label{eq:condi_lemma/asrate}
			\frac{1}{2p}\sum_{t=k-p+1}^{k} \mathtt{U}_t + \mathtt{U}_t\tr \geq a I_n
		\end{equation}
		for some $ p \in \mathbb{N} $, $ a > 0 $ and all $ k \in \mathbb{N} $ almost surely;}
		\item $ \{\mathtt{W}_k,\mathcal{F}_k\} $ is a martingale difference sequence such that $ \mE \left[\!\Absl{\mathtt{W}_k}^{\rho} \middle| \mathcal{F}_{k-1}\! \right] \! =  O \! \left(\! \frac{1}{k^{\rho h}}\! \right) $ almost surely for some $ \rho > 2 $ and $ \frac{1}{2} < h \leq \min\{1,\frac{3+2h-2\mu}{4}\} $;		
		\item $ \{\mathtt{X}_k,\mathcal{F}_k\} $ is a sequence of adaptive random variables;
		\item There exists $ c > 1 $ almost surely such that
		\begin{equation}\label{eq:lemma_condi/as_rate}
			\mathtt{X}_k = \left(I_n - \frac{\mathtt{U}_k}{k} + O\left(\frac{1}{k^{c}}\right)\right) \mathtt{X}_{k-1} + \mathtt{W}_k. 
		\end{equation}
	\end{enumerate}
	Then,
	\vspace{-2pt}
	\begin{align*}
		\mathtt{X}_k = 
		\begin{cases}
			O\left( \frac{1}{k^a} \right), & \text{if}\ 2h-2a>1;\\
			O\left( \frac{\ln k}{k^{h-1/2}} \right), & \text{if}\ 2h-2a=1;\\
			O\left( \frac{\sqrt{\ln k}}{k^{h-1/2}} \right), & \text{if}\ 2h-2a<1,
		\end{cases}\ \as
	\end{align*}
\end{lemma}

\begin{proof}
	\MODIFY{Denote $ \bar{\mathtt{U}}_t = \frac{\mathtt{U}_t + \mathtt{U}_t\tr}{2} $.} Then, by \eqref{eq:lemma_condi/as_rate}, 
	\begin{align*}
		& \mE\left[ \Abs{\mathtt{X}_k}^2 \middle| \mathcal{F}_{k-1} \right] \\
		= & \left( 1 + O\left(\frac{1}{k^{\min\{c,2-2\mu\}}}\right) \right)\Abs{\mathtt{X}_{k-1}}^2  - \frac{2}{k} \mathtt{X}_{k-1}\tr \bar{\mathtt{U}}_k \mathtt{X}_{k-1} + O\left( \frac{1}{k^{2b}} \right),\ \as
	\end{align*}
	Hence, by Theorem 1.3.2 of \cite{GuoL2020}, we have $ \Abs{\mathtt{X}_k}^2 $ converges to a finite value almost surely, which implies the almost sure boundedness of $ \mathtt{X}_k $.  
	
	We estimate the almost sure convergence rate of $ \mathtt{X}_k $ in the following two cases. 
	
	\textbf{Case 1: $ 2h - 2a > 1 $.}
	In this case, we have
	\begin{align}\label{eq:condE_rate_case1}
		& \mE\left[ (k+1)^{2a}\Abs{\mathtt{X}_k}^2 \middle| \mathcal{F}_{k-1} \right] \\
		\leq \ \ \ \ \ \ & \!\!\!\!\!\!\!\!\!\!\!\!\! \left(1 + \frac{2a}{k} + O\left(\frac{1}{k^{\min\{c,2-2\mu\}}}\right)\right) k^{2a} \Abs{\mathtt{X}_{k-1}}^2 - \frac{2}{k^{1-2a}} \mathtt{X}_{k-1}\tr \bar{\mathtt{U}}_k \mathtt{X}_{k-1} + O\left( \frac{1}{k^{2h-2a}} \right) \nonumber\\
		= \ \ \ \ \ \ & \!\!\!\!\!\!\!\!\!\!\!\!\! \left(1 + O\left(\frac{1}{k^{\min\{c,2-2\mu\}}}\right)\right) k^{2a} \Abs{\mathtt{X}_{k-1}}^2 
		+ \frac{2a}{k^{1-2a}} \Abs{\mathtt{X}_{k-1}}^2 - \frac{2}{k^{1-2a}} \mathtt{X}_{k-1}\tr \bar{\mathtt{U}}_k \mathtt{X}_{k-1} \nonumber\\
		\ \ \ \ \ \ & \!\!\!\!\!\!\!\!\!\!\!\!\! + O\left( \frac{1}{k^{2h-2a}} \right),\ \as \nonumber
	\end{align}
	Next, we will prove that $ \sup_{k\in\mathbb{N}}\sum_{t=1}^{k}\left(\frac{2a}{t^{1-2a}} \Abs{\mathtt{X}_{t-1}}^2 - \frac{2}{t^{1-2a}} \mathtt{X}_{t-1}\tr \bar{\mathtt{U}}_t \mathtt{X}_{t-1}\right)\! <\! \infty $ almost surely. 
	Note that $ 1 - 2a > 2 - 2h \geq 0 $. 
	Then, by \eqref{eq:condi_lemma/asrate}, one can get
	\begin{align}\label{ineq:sum_minus}
		& \!\!\!\!\! \!\!\!  \sum_{t=1}^{k} \left(\frac{2a}{t^{1-2a}} \Abs{\mathtt{X}_{t-1}}^2 - \frac{2}{t^{1-2a}} \mathtt{X}_{t-1}\tr \bar{\mathtt{U}}_t \mathtt{X}_{t-1}\right) \\
		\leq \sum_{r=0}^{\lfloor \frac{k}{p} \rfloor - 1} & \sum_{t=pr+1}^{pr+p} \left( \frac{2a}{t^{1-2a}}\Abs{\mathtt{X}_{t-1}}^2 - \frac{2}{t^{1-2a}} \mathtt{X}_{t-1}\tr \bar{\mathtt{U}}_t \mathtt{X}_{t-1} \right) + O\left( \frac{1}{k^{1-2a}} \right) \nonumber \\
		\leq \sum_{r=0}^{\lfloor \frac{k}{p} \rfloor - 1} & \!\! \frac{2}{(pr+p)^{1-2a}} \! \sum_{t=pr+1}^{pr+p} \!\!  \left(a \Abs{\mathtt{X}_{t-1}}^2 - \mathtt{X}_{t-1}\tr \bar{\mathtt{U}}_t \mathtt{X}_{t-1} \right) + \!\! \sum_{r=0}^{\lfloor \frac{k}{p} \rfloor - 1} \!\! O\left( \frac{1}{r^{2-2a}} \right) + O\left( 1 \right) \nonumber\\
		\leq \sum_{r=0}^{\lfloor \frac{k}{p} \rfloor - 1} & \frac{2}{(pr+p)^{1-2a}} \sum_{t=pr+1}^{pr+p}  \left(\mathtt{X}_{pr+p-1}\tr \bar{\mathtt{U}}_t \mathtt{X}_{pr+p-1} - \mathtt{X}_{t-1}\tr \bar{\mathtt{U}}_t \mathtt{X}_{t-1} \right) \nonumber\\
		+ \sum_{r=0}^{\lfloor \frac{k}{p} \rfloor - 1}\!\!\!\!\!\! & \ \ \ \ \frac{2a}{(pr+p)^{1-2a}} \sum_{t=pr+1}^{pr+p}  \left(\Abs{\mathtt{X}_{t-1}}^2 - \Abs{\mathtt{X}_{pr+p-1}}^2 \right) + O\left( 1 \right). \nonumber
	\end{align} 
	Besides,
	\begin{align}\label{eq:sum_minus1}
		& \!\!\!\!\! \sum_{t=pr+1}^{pr+p} \left( \mathtt{X}_{pr+p-1}\tr \bar{\mathtt{U}}_t \mathtt{X}_{pr+p-1} - \mathtt{X}_{t-1}\tr \bar{\mathtt{U}}_t \mathtt{X}_{t-1} \right) \\
		= \ \ \ \ \ \ \ \ & \!\!\!\!\!\!\!\!\!\!\!\!\!\!\!\!\! \sum_{t=pr+1}^{pr+p} \sum_{l=t}^{pr+p-1} \left( \mathtt{X}_{l}\tr \bar{\mathtt{U}}_t \mathtt{X}_{l} - \mathtt{X}_{l-1}\tr \bar{\mathtt{U}}_t \mathtt{X}_{l-1} \right) \nonumber\\
		= \ \ \ \ \ \ \ \ & \!\!\!\!\!\!\!\!\!\!\!\!\!\!\!\!\! \sum_{t=pr+1}^{pr+p} \sum_{l=t}^{pr+p-1} \left( 2 \mathtt{W}_l\tr \bar{\mathtt{U}}_t \left( I_n - \frac{\bar{\mathtt{U}}_l}{l} + O\left( \frac{1}{l^c} \right) \right) \mathtt{X}_{l-1} + \mathtt{W}_l\tr \bar{\mathtt{U}}_t \mathtt{W}_l \right) +O\left( r^{2\mu-1} \right),\ \as \nonumber  
	\end{align}
	When $ t \in \left\{pq + 1, \ldots, pq + l \right\} $ and $ l = \{ t, \ldots, pr+p-1 \} $, it holds that
	\begin{align*}
		\frac{4}{(pr+p)^{1-2a}l^b} \bar{\mathtt{U}}_t \left( I_n - \frac{\bar{\mathtt{U}}_l}{l} + O\left( \frac{1}{l^c} \right) \right) \mathtt{X}_{l-1} = O\left( \frac{1}{r^{1+b-2a-\mu}} \right). 
	\end{align*}
	Note that $ 1 + h - 2a- \mu \geq 2h - 2a - \mu > \frac{1}{2} $. Then, by Theorem 1.3.10 of \cite{GuoL2020}, we have
	\begin{align}\label{eq:sum3_asbounded}
		\sum_{r=0}^{\lfloor \frac{k}{p} \rfloor - 1} \!\! \sum_{t=pr+1}^{pr+p}\!\! \sum_{l=t}^{pr+p-1} \!\! (l^b \mathtt{W}_l)\tr \! \left( \frac{4}{(pr+p)^{1-2a}l^b} \bar{\mathtt{U}}_t \! \left(\! I_n - \frac{\bar{\mathtt{U}}_l}{l} + O\left( \frac{1}{l^c} \right)\! \right)\! \mathtt{X}_{l-1}\! \right) \!\! = O(1),\ \as
	\end{align}
	Additionally, by $ 1 + 2b - 2a - \mu > 2 - \mu > 1 $ and Theorem 1.3.9 of \cite{GuoL2020} with $ \alpha = 1 $,
	\begin{align*}
		& \sum_{r=0}^{\lfloor \frac{k}{p} \rfloor - 1}  \sum_{t=pr+1}^{pr+p} \frac{2}{(pr+p)^{1-2a}} \sum_{l=t}^{pr+p-1} \mathtt{W}_l\tr \bar{\mathtt{U}}_t \mathtt{W}_l \\
		= & \sum_{r=0}^{\lfloor \frac{k}{p} \rfloor - 1}  \sum_{t=pr+1}^{pr+p} \sum_{l=t}^{pr+p-1} \frac{2}{(pr+p)^{1-2a} t^{2b-\mu}} \cdot t^{2b-\mu} \left(\mathtt{W}_l\tr \bar{\mathtt{U}}_t \mathtt{W}_l -\mE\left[ \mathtt{W}_l\tr \bar{\mathtt{U}}_t \mathtt{W}_l \middle| \mathcal{F}_{l-1} \right]  \right) \\
		& + \sum_{r=0}^{\lfloor \frac{k}{p} \rfloor - 1}  \sum_{t=pr+1}^{pr+p} \frac{2}{(pr+p)^{1-2a}} \sum_{l=t}^{pr+p-1} \mE\left[ \mathtt{W}_l\tr \bar{\mathtt{U}}_t \mathtt{W}_l \middle| \mathcal{F}_{l-1} \right] 
		= O(1),\ \as,
	\end{align*}
	which together with \eqref{eq:sum_minus1} and \eqref{eq:sum3_asbounded} implies that 
	\begin{align*}
		\sum_{r=0}^{\lfloor \frac{k}{p} \rfloor - 1} \frac{2}{(pr+p)^{1-2a}} \sum_{t=pr+1}^{pr+p}  \left(\mathtt{X}_{pr+p-1}\tr \bar{\mathtt{U}}_t \mathtt{X}_{pr+p-1} - \mathtt{X}_{t-1}\tr \bar{\mathtt{U}}_t \mathtt{X}_{t-1} \right) = O(1),\ \as
	\end{align*}
	Similarly, one can get
	\begin{align*}
		\sum_{r=0}^{\lfloor \frac{k}{p} \rfloor - 1} \frac{2a}{(pr+p)^{1-2a}} \sum_{t=pr+1}^{pr+p}  \left(\Abs{\mathtt{X}_{t-1}}^2 - \Abs{\mathtt{X}_{pr+p-1}}^2 \right) = O(1),\ \as
	\end{align*}
	Then, by \eqref{ineq:sum_minus}, we have 
	\begin{align}\label{ineq:sum_minus_asbounded}
		\sum_{t=1}^{k} \left(\frac{2a}{t^{1-2a}} \Abs{\mathtt{X}_{t-1}}^2 - \frac{2}{t^{1-2a}} \mathtt{X}_{t-1}\tr \bar{\mathtt{U}}_t \mathtt{X}_{t-1}\right) < \infty, \ \as,
	\end{align}
	
	Given $ S_0 > 0$, define $ \mathtt{S}_k = S_0 -  \sum_{t=1}^{k} \left(\frac{2a}{t^{1-2a}} \Abs{\mathtt{X}_{t-1}}^2 - \frac{2}{t^{1-2a}} \mathtt{X}_{t-1}\tr \bar{\mathtt{U}}_t \mathtt{X}_{t-1}\right) $ and $ \mathtt{V}_k = (k+1)^{2a}\Abs{\mathtt{X}_k}^2 + \mathtt{S}_k $. Hence by \eqref{eq:condE_rate_case1}, we have
	\begin{align*}
		\mE\left[ \mathtt{V}_k \middle| \mathcal{F}_{k-1} \right]
		\leq \left(1 + O\left(\frac{1}{k^{\min\{c,2-2\mu\}}}\right)\right) \mathtt{V}_{k-1} + O\left( \frac{1}{k^{2h-2a}} \right),\ \as \nonumber
	\end{align*}
	Then, define $ \mathtt{k}_0 = \inf \{ k : \mathtt{S}_k < 0 \} $. We have
	\begin{align*}
		& \mE\left[ \mathtt{V}_{\min\{k,\mathtt{k}_0\}} \middle| \mathcal{F}_{k-1} \right] \\
		\leq & \mathtt{V}_{\mathtt{k}_0} I_{\{\mathtt{k}_0 \leq k\}} + \left(1 + O\left(\frac{1}{k^{\min\{c,2-2\mu\}}}\right)\right) \mathtt{V}_{k-1}I_{\{\mathtt{k}_0 > k\}} + O\left( \frac{1}{k^{2h-2a}} \right) \\
		\leq & \left(1 + O\left(\frac{1}{k^{\min\{c,2-2\mu\}}}\right)\right) \mathtt{V}_{\min\{k-1,\mathtt{k}_0\}} + O\left( \frac{1}{k^{2h-2a}} \right). 
	\end{align*}
	By Theorem 1.3.2 of \cite{GuoL2020}, $ \mathtt{V}_{\min\{k,\mathtt{k}_0\}} $ converges to a finite value almost surely. Note that $ \mathtt{V}_k = \mathtt{V}_{\min\{k,\mathtt{k}_0\}} $ in the set
	\begin{align*}
		\{ \mathtt{k}_0 = \infty \} 
		= \{ \inf_k \mathtt{S}_k \geq 0 \}
		= \left\{ \sum_{t=1}^{k} \left(\frac{2a}{t^{1-2a}} \Abs{\mathtt{X}_{t-1}}^2 - \frac{2}{t^{1-2a}} \mathtt{X}_{t-1}\tr \bar{\mathtt{U}}_t \mathtt{X}_{t-1}\right) < S_0 \right\}.
	\end{align*}
	Then, by the arbitrariness of $ S_0 $ and \eqref{ineq:sum_minus_asbounded},  $ \mathtt{V}_k $ converges to a finite value almost surely, which implies the almost sure boundedness of $ (k+1)^{2a}\Abs{\mathtt{X}_k}^2 $. Hence, one can get $ \mathtt{X}_k = O\left( \frac{1}{k^a} \right) $ almost surely. 
	
	\textbf{Case 2: $ 2h - 2a \leq 1 $. } In this case, we have
	\begin{align}\label{eq:condE_rate_case2}
		& \mE\left[ \frac{(k+1)^{2h-1}}{(\ln (k+1))^2}\Abs{\mathtt{X}_k}^2 \middle| \mathcal{F}_{k-1} \right] \\
		\leq &  \left( 1 + \frac{2h-1}{k} + O\left(\frac{1}{k^{\min\{c,2-2\mu\}}}\right) \right) \frac{k^{2h-1}}{(\ln k)^2} \Abs{\mathtt{X}_{k-1}}^2   \nonumber\\
		& - \frac{2}{k^{2-2h}(\ln k)^2} \mathtt{X}_{k-1}\tr \bar{\mathtt{U}}_k \mathtt{X}_{k-1} + O\left( \frac{1}{k (\ln k)^2} \right) \nonumber \\
		\leq & \left( 1 + O\left(\frac{1}{k^{\min\{c,2-2\mu\}}}\right) \right) \frac{k^{2h-1}}{(\ln k)^2} \Abs{\mathtt{X}_{k-1}}^2 + \frac{2a}{k^{2-2h}(\ln k)^2} \Abs{\mathtt{X}_{k-1}}^2    \nonumber\\
		& - \frac{2}{k^{2-2h}(\ln k)^2} \mathtt{X}_{k-1}\tr \bar{\mathtt{U}}_k \mathtt{X}_{k-1} + O\left( \frac{1}{k (\ln k)^2} \right),\ \as \nonumber
	\end{align}
	Then, similar to the case of $ 2h - 2a > 1 $, we have $ \mathtt{X}_k = O\left( \frac{\ln k}{k^{h-1/2}} \right) $ almost surely. 
	
	We further promote the almost sure convergence rate for the case of $ 2h - 2a < 1 $. Since  $ \mathtt{X}_k = O\left( \frac{\ln k}{k^{h-1/2}} \right) $ almost surely, one can get
	\begin{align}\label{ineq:iter_case3}
		& \!\! (k+1)^{2h-1} \Abs{\mathtt{X}_k}^2 \\
		\leq \ \ \ \ \ \ \ \ & \!\!\!\!\!\!\!\!\!\!\!\!\!\!\! k^{2h-1} \Abs{\mathtt{X}_{k-1}}^2 + 2 (k+1)^{2h-1} \mathtt{W}_k\tr \left(\! I_n - \frac{\bar{\mathtt{U}}_k}{k} + O\left( \frac{1}{k^c} \right)\! \right) \mathtt{X}_{k-1}  + \frac{2h -1 }{k^{2-2h}} \Abs{\mathtt{X}_{k-1}}^2 \nonumber\\
		\ \ \ \ \ \ \ \ & \!\!\!\!\!\!\!\!\!\!\!\!\!\!\! - \frac{2}{k^{2-2h}} \mathtt{X}_{k-1}\tr \bar{\mathtt{U}}_k \mathtt{X}_{k-1} + (k+1)^{2h-1} \left( \Absl{\mathtt{W}_k}^2 - \mE \left[ \Absl{\mathtt{W}_k}^2 \middle | \mathcal{F}_{k-1} \right] \right) + O\left( \frac{1}{k} \right). \nonumber
	\end{align}
	By Theorem 1.3.10 of \cite{GuoL2020}, it holds that
	\begin{align*}
		& \sum_{t=1}^{k} 2 (t+1)^{2h-1} \mathtt{W}_t\tr \left( I_n - \frac{\bar{\mathtt{U}}_t}{t} + O\left( \frac{1}{t^c} \right) \right) \mathtt{X}_{t-1} \\
		= & \sum_{t=1}^{k}  ( (t+1)^h \mathtt{W}_t)\tr \left(2 (t+1)^{h-1} \left( I_n - \frac{\bar{\mathtt{U}}_t}{t} + O\left( \frac{1}{t^c} \right) \right) \mathtt{X}_{t-1}\right) \\
		= & O(1) + o\left( \sum_{t=1}^{k} \frac{1}{t^{2-2h} } \Absl{\mathtt{X}_t}^2\right),\ \as
	\end{align*}
	\MODIFY{By Theorem 1.3.9 of \cite{GuoL2020} with $ \alpha = 1 $}, one can get 
	\begin{align*}
		& \sum_{t=1}^{k} (t+1)^{2h-1} \left( \Absl{\mathtt{W}_t}^2 - \mE \left[ \Absl{\mathtt{W}_t}^2 \middle | \mathcal{F}_{k-1} \right] \right) \\
		= & \sum_{t=1}^{k} (t+1)^{2h} \left( \Absl{\mathtt{W}_t}^2 - \mE \left[ \Absl{\mathtt{W}_t}^2 \middle | \mathcal{F}_{k-1} \right] \right) \cdot \frac{1}{t+1} = O(\ln k),\ \as
	\end{align*}
	Similar to \eqref{ineq:sum_minus_asbounded}, we have
	\begin{align*}
		\sum_{t=1}^{k} \left(\frac{2a}{t^{2-2h}} \Abs{\mathtt{X}_{t-1}}^2 - \frac{2}{t^{2-2h}} \mathtt{X}_{t-1}\tr \bar{\mathtt{U}}_t \mathtt{X}_{t-1}\right) \leq  o(\ln k),\ \as
	\end{align*}
	Hence, by \eqref{ineq:iter_case3},
	\begin{align*}
		& (k+1)^{2h-1} \Abs{\mathtt{X}_k}^2 \\
		\leq & \Abs{\mathtt{X}_0}^2 + \sum_{t=1}^{k} 2 (t+1)^{2h-1} \mathtt{W}_t\tr \left( I_n - \frac{\bar{\mathtt{U}}_t}{t} + O\left( \frac{1}{t^c} \right) \right) \mathtt{X}_{t-1} \\
		& - \sum_{t=1}^{k} \frac{1+2a-2h}{t^{2-2h}} \Abs{\mathtt{X}_{t-1}}^2 + \sum_{t=1}^{k} \left(\frac{2a}{t^{2-2h}} \Abs{\mathtt{X}_{t-1}}^2 - \frac{2}{t^{2-2h}} \mathtt{X}_{t-1}\tr \bar{\mathtt{U}}_t \mathtt{X}_{t-1}\right) \\
		&  + \sum_{t=1}^{k} (t+1)^{2h-1} \left( \Absl{\mathtt{W}_t}^2 - \mE \left[ \Absl{\mathtt{W}_t}^2 \middle | \mathcal{F}_{k-1} \right] \right) + O(\ln k) \\
		\leq & o\left( \sum_{t=1}^{k} \frac{1}{t^{2-2h}} \Abs{\mathtt{X}_t}^2 \right) - \left( 1+2a-2h \right) \sum_{t=1}^{k} \frac{1}{t^{2-2h}} \Abs{\mathtt{X}_t}^2 + O(\ln k)
		= O(\ln k),\ \as,
	\end{align*} 
	which implies $ \mathtt{X}_k = O\left(\frac{\sqrt{\ln k}}{k^{h-1/2}}\right) $. 
	The lemma is thereby proved. 
\end{proof}


\begin{thebibliography}{10}

\bibitem{aysal2008distributed}
{\sc T.~C. Aysal, M.~J. Coates, and M.~G. Rabbat}, {\em Distributed average
  consensus with dithered quantization}, IEEE transactions on Signal
  Processing, 56 (2008), pp.~4905--4918.

\bibitem{carli2009quantized}
{\sc R.~Carli and F.~Bullo}, {\em Quantized coordination algorithms for
  rendezvous and deployment}, SIAM Journal on Control and Optimization, 48
  (2009), pp.~1251--1274.

\bibitem{carli2008communication}
{\sc R.~Carli, F.~Fagnani, A.~Speranzon, and S.~Zampieri}, {\em Communication
  constraints in the average consensus problem}, Automatica, 44 (2008),
  pp.~671--684.

\bibitem{chang2023fully}
{\sc Z.~Chang, W.~Song, J.~Wang, and Z.~Li}, {\em Fully distributed
  event-triggered affine formation maneuver control over directed graphs},
  Science China Information Sciences, 66 (2023), pp.~1--3.

\bibitem{chen2022distributed}
{\sc G.~Chen, D.~Yao, Q.~Zhou, H.~Li, and R.~Lu}, {\em Distributed
  event-triggered formation control of {USVs} with prescribed performance},
  Journal of Systems Science and Complexity,  (2022), p.~820–838.

\bibitem{gan2022distributed}
{\sc D.~Gan and Z.~Liu}, {\em Distributed order estimation of {ARX} model under
  cooperative excitation condition}, SIAM Journal on Control and Optimization,
  60 (2022), pp.~1519--1545.

\bibitem{he2022event}
{\sc X.~He, Y.~Xing, J.~Wu, and K.~H. Johansson}, {\em Event-triggered
  distributed estimation with decaying communication rate}, SIAM Journal on
  Control and Optimization, 60 (2022), pp.~992--1017.

\bibitem{kar2014distributed}
{\sc S.~Kar, G.~Hug, J.~Mohammadi, and J.~M. Moura}, {\em Distributed state
  estimation and energy management in smart grids: A consensus $+$ innovations
  approach}, IEEE Journal of selected topics in signal processing, 8 (2014),
  pp.~1022--1038.

\bibitem{kar2013distributed}
{\sc S.~Kar, J.~M. Moura, and H.~V. Poor}, {\em Distributed linear parameter
  estimation: Asymptotically efficient adaptive strategies}, SIAM Journal on
  Control and Optimization, 51 (2013), pp.~2200--2229.

\bibitem{kar2012distributed}
{\sc S.~Kar, J.~M. Moura, and K.~Ramanan}, {\em Distributed parameter
  estimation in sensor networks: Nonlinear observation models and imperfect
  communication}, IEEE Transactions on Information Theory, 58 (2012),
  pp.~3575--3605.

\bibitem{Ke2023Signal}
{\sc J.~M. Ke, Y.~L. Zhao, and J.~F. Zhang}, {\em Signal comparison average
  consensus algorithm under binary-valued communications}, in Proceedings of
  the IEEE Conference on Decision and Control, Singapore, December 2023.

\bibitem{li2009distributed}
{\sc J.~Li and G.~AlRegib}, {\em Distributed estimation in energy-constrained
  wireless sensor networks}, IEEE Transactions on Signal Processing, 57 (2009),
  pp.~3746--3758.

\bibitem{LiT2011distributed}
{\sc T.~Li, M.~Fu, L.~Xie, and J.-F. Zhang}, {\em Distributed consensus with
  limited communication data rate}, IEEE Transactions on Automatic Control, 56
  (2011), pp.~279--292.

\bibitem{liu2011distributed}
{\sc S.~Liu, T.~Li, and L.~Xie}, {\em Distributed consensus for multiagent
  systems with communication delays and limited data rate}, SIAM Journal on
  Control and Optimization, 49 (2011), pp.~2239--2262.

\bibitem{meng2017finite}
{\sc Y.~Meng, T.~Li, and J.-F. Zhang}, {\em Finite-level quantized
  synchronization of discrete-time linear multiagent systems with switching
  topologies}, SIAM Journal on Control and Optimization, 55 (2017),
  pp.~275--299.

\bibitem{robbins1971convergence}
{\sc H.~Robbins and D.~Siegmund}, {\em A convergence theorem for non negative
  almost supermartingales and some applications}, in Optimizing methods in
  statistics, Elsevier, 1971, pp.~233--257.

\bibitem{sahu2018CIRFE}
{\sc A.~K. Sahu, D.~Jakoveti{\'c}, and S.~Kar}, {\em $\mathcal{CIRFE}$: A
  distributed random fields estimator}, IEEE Transactions on Signal Processing,
  66 (2018), pp.~4980--4995.

\bibitem{sayed2014diffusion}
{\sc A.~H. Sayed}, {\em Diffusion adaptation over networks}, in Academic Press
  Library in Signal Processing, vol.~3, Elsevier, 2014, pp.~323--453.

\bibitem{sayin2013single}
{\sc M.~O. Sayin and S.~S. Kozat}, {\em Single bit and reduced dimension
  diffusion strategies over distributed networks}, IEEE Signal Processing
  Letters, 20 (2013), pp.~976--979.

\bibitem{sayin2014compressive}
{\sc M.~O. Sayin and S.~S. Kozat}, {\em Compressive diffusion strategies over
  distributed networks for reduced communication load}, IEEE Transactions on
  Signal Processing, 62 (2014), pp.~5308--5323.

\bibitem{Shiryaev}
{\sc A.~N. Shiryaev}, {\em Probability}, Graduate texts in mathematics Springer
  series in Soviet mathematics 95, Springer, 2nd ed~ed., 1996.

\bibitem{stout1970martingale}
{\sc W.~F. Stout}, {\em A martingale analogue of {K}olmogorov's law of the
  iterated logarithm}, Zeitschrift f{\"u}r Wahrscheinlichkeitstheorie und
  verwandte Gebiete, 15 (1970), pp.~279--290.

\bibitem{LiT2021decentralized}
{\sc J.~Wang, T.~Li, and X.~Zhang}, {\em Decentralized cooperative online
  estimation with random observation matrices, communication graphs and time
  delays}, IEEE Transactions on Information Theory, 67 (2021), pp.~4035--4059.

\bibitem{wang2023differentially}
{\sc J.~Wang, J.~Tan, and J.-F. Zhang}, {\em Differentially private distributed
  parameter estimation}, Journal of Systems Science and Complexity, 36 (2023),
  pp.~187--204.

\bibitem{wang2019consensus}
{\sc T.~Wang, H.~Zhang, and Y.~Zhao}, {\em Consensus of multi-agent systems
  under binary-valued measurements and recursive projection algorithm}, IEEE
  Transactions on Automatic Control, 65 (2019), pp.~2678--2685.

\bibitem{wang2023event}
{\sc Z.~Wang, C.~Jin, W.~He, M.~Xiao, G.-P. Jiang, and J.~D. Cao}, {\em
  Event-triggered impulsive synchronization of heterogeneous neural networks},
  Science China Information Sciences,  (2023).
\newblock
  \href{https://doi.org/10.1007/s11432-022-3839-y}{DOI:10.1007/s11432-022-3839-y}.

\bibitem{wei1985asymptotic}
{\sc C.-Z. Wei}, {\em Asymptotic properties of least-squares estimates in
  stochastic regression models}, The Annals of Statistics,  (1985),
  pp.~1498--1508.

\bibitem{wu2022dynamic}
{\sc X.~Wu, B.~Mao, X.~Wu, and J.~L{\"u}}, {\em Dynamic event-triggered
  leader-follower consensus control for multiagent systems}, SIAM Journal on
  Control and Optimization, 60 (2022), pp.~189--209.

\bibitem{xie2018analysis}
{\sc S.~Xie and L.~Guo}, {\em Analysis of normalized least mean squares-based
  consensus adaptive filters under a general information condition}, SIAM
  Journal on Control and Optimization, 56 (2018), pp.~3404--3431.

\bibitem{zhao2018consensus}
{\sc Y.~Zhao, T.~Wang, and W.~Bi}, {\em Consensus protocol for multiagent
  systems with undirected topologies and binary-valued communications}, IEEE
  Transactions on Automatic Control, 64 (2018), pp.~206--221.

\bibitem{zhu2018mitigating}
{\sc S.~Zhu, C.~Chen, J.~Xu, X.~Guan, L.~Xie, and K.~H. Johansson}, {\em
  Mitigating quantization effects on distributed sensor fusion: A least squares
  approach}, IEEE Transactions on Signal Processing, 66 (2018), pp.~3459--3474.

\bibitem{zhu2015distributed}
{\sc S.~Zhu, Y.~C. Soh, and L.~Xie}, {\em Distributed parameter estimation with
  quantized communication via running average}, IEEE Transactions on Signal
  Processing, 63 (2015), pp.~4634--4646.

\bibitem{zorich}
{\sc V.~A. Zorich}, {\em Mathematical Analysis I}, Springer.

\end{thebibliography}


\begin{thebibliography}{10}
	\bibitem{alistarh2017QSGD}
	{\sc D. Alistarh, D. Grubic, J. Li, R. Tomioka, and M. Vojnovic}, {\em QSGD: Communication-efficient SGD via gradient quantization and encoding,}
	in Proceedings of Advances in Neural Information Processing Systems, Long Beach, CA, 2017.
	
	\bibitem{aysal2008distributed}
	{\sc T.~C. Aysal, M.~J. Coates, and M.~G. Rabbat}, {\em Distributed average consensus with dithered quantization}, IEEE Trans. Signal Process., 56 (2008), pp.~4905--4918.
	
	\bibitem{carli2009quantized}
	{\sc R.~Carli and F.~Bullo}, {\em Quantized coordination algorithms for	rendezvous and deployment}, SIAM J. Control Optim., 48 (2009), pp.~1251--1274.
	
	\bibitem{carli2008communication}
	{\sc R.~Carli, F.~Fagnani, A.~Speranzon, and S.~Zampieri}, {\em Communication constraints in the average consensus problem}, Automatica, 44 (2008),
	pp.~671--684.
	
	\bibitem{carpentiero2021distributed}
	{\sc M. Carpentiero, V. Matta, and A. H. Sayed}, {\em Distributed Adaptive Learning Under Communication Constraints},  IEEE Open Journal of Signal Processing, 5 (2023), pp.~321--358.
	
	

	\bibitem{ChenHF2003SA}
	{\sc H. F. Chen}, {\em Stochastic Approximation and Its
	Applications}, Kluwer Academic Publishers, Dordrecht, 2003.	
	
	\bibitem{gan2022distributed}
	{\sc D.~Gan and Z.~Liu}, {\em Distributed order estimation of {ARX} model under cooperative excitation condition}, SIAM J. Control Optim.,	60 (2022), pp.~1519--1545.
	
	\bibitem{GuoL2020}
		{\sc L. Guo}, {\em Time-Varying Stochastic Systems: Stability and Adaptive Theory}, 2nd ed., Science Press, Beijing, 2020.
	
	\bibitem{Gustafsson2013generating}
	{\sc F. Gustafsson and R. Karlsson}, {\em Generating dithering noise for maximum likelihood estimation from quantized data}, Automatica, 49 (2013), pp.~554–-560. 
	
	\bibitem{he2022event}
	{\sc X.~He, Y.~Xing, J.~Wu, and K.~H. Johansson}, {\em Event-triggered distributed estimation with decaying communication rate}, SIAM J. Control Optim., 60 (2022), pp.~992--1017.
	
	\bibitem{kar2014distributed}
	{\sc S.~Kar, G.~Hug, J.~Mohammadi, and J.~M. Moura}, {\em Distributed state estimation and energy management in smart grids: A consensus$ + $innovations approach}, IEEE J. Sel. Top. Signal Process., 8 (2014),
	pp.~1022--1038.
	
	\bibitem{kar2013distributed}
	{\sc S.~Kar, J.~M. Moura, and H.~V. Poor}, {\em Distributed linear parameter estimation: Asymptotically efficient adaptive strategies}, SIAM J. Control Optim., 51 (2013), pp.~2200--2229.
	
	\bibitem{kar2012distributed}
	{\sc S.~Kar, J.~M. Moura, and K.~Ramanan}, {\em Distributed parameter estimation in sensor networks: Nonlinear observation models and imperfect communication}, IEEE Trans. Inform. Theory, 58 (2012),
	pp.~3575--3605.
	
	\bibitem{Ke2023Signal}
	{\sc J.~M. Ke, Y.~L. Zhao, and J.~F. Zhang}, {\em Signal comparison average consensus algorithm under binary-valued communications}, in Proceedings of the IEEE Conference on Decision and Control, Singapore, 2023.
	
	\bibitem{lao2022quantized}
	{\sc X. Lao, W. Du, and C. Li}, {\em Quantized distributed estimation with adaptive combiner}, IEEE Trans. Signal Inf. Process. Netw., 8 (2022), pp. 187--200.
	
	\bibitem{li2009distributed}
	{\sc J.~Li and G.~AlRegib}, {\em Distributed estimation in energy-constrained wireless sensor networks}, IEEE Trans. Signal Process., 57 (2009),
	pp.~3746--3758.
	
	\bibitem{LiT2011distributed}
	{\sc T.~Li, M.~Fu, L.~Xie, and J. F. Zhang}, {\em Distributed consensus with limited communication data rate}, IEEE Trans. Automat. Control, 56 (2011), pp.~279--292.
	
	\bibitem{liu2011distributed}
	{\sc S.~Liu, T.~Li, and L.~Xie}, {\em Distributed consensus for multiagent systems with communication delays and limited data rate}, SIAM J. Control Optim., 49 (2011), pp.~2239--2262.
	
	\bibitem{Jakovetic2023heavy-tail}
	{\sc D. Jakovetic, M. Vukovic, D. Bajovic, A. K. Sahu, and S. Kar}, {\em Distributed recursive estimation under heavy-tail communication noise}, SIAM J. Control Optim., 61 (2023), pp.~1582--1609.
	
	\bibitem{meng2017finite}
	{\sc Y.~Meng, T.~Li, and J. F. Zhang}, {\em Finite-level quantized
		synchronization of discrete-time linear multiagent systems with switching topologies}, SIAM J. Control Optim., 55 (2017),
	pp.~275--299.
	
	\MODIFY{
	\bibitem{Michelusi2022}	
	{\sc N. Michelusi, G. Scutari, and C. S. Lee}, {\em Finite-bit quantization for distributed algorithms with linear convergence}, IEEE Trans. Inform. Theory, 68 (2022), pp.~ 7254--7280.

	\bibitem{Nassif2023}
	{\sc R. Nassif, S. Vlaski, M. Carpentiero, V. Matta, M. Antonini, and A. H. Sayed}, {\em Quantization for decentralized learning under subspace constraints}, IEEE Trans. Signal Process., 71 (2023), pp.~2320--2335. 
}
	
	
	\bibitem{sahu2018CIRFE}
	{\sc A.~K. Sahu, D.~Jakoveti{\'c}, and S.~Kar}, {\em $\mathcal{CIRFE}$: A
		distributed random fields estimator}, IEEE Trans. Signal Process.,
	66 (2018), pp.~4980--4995.
	
	\bibitem{sayed2014diffusion}
	{\sc A.~H. Sayed}, {\em Diffusion adaptation over networks}, 
	in Academic Press Library in Signal Processing, R. Chellapa and S. Theodoridis, eds., 2014, pp. 323--454.	
	
	
	\bibitem{sayin2014compressive}
	{\sc M.~O. Sayin and S.~S. Kozat}, {\em Compressive diffusion strategies over
		distributed networks for reduced communication load}, IEEE Trans. Signal Process., 62 (2014), pp.~5308--5323.
	
	\bibitem{Shiryaev}
	{\sc A.~N. Shiryaev}, {\em Probability}, 2nd ed., Springer, New York, 1996. 
	
	\bibitem{YuanS2024JSSC}
	{\sc S. Yuan, C. P. Yu, and J. Sun}, {\em Hybrid adaptive event-triggered consensus control with intermittent communication and control updating}, J. Syst. Sci. Complex., (2024), https://doi.org/10.1007/s11424-024-3563-8. 
	
	\bibitem{Yang2024SCIS}
	{\sc S. Yang, W. Xu, W. He, and J. Cao}, {\em Distributed generalized Nash equilibrium seeking: event-triggered coding-decoding-based secure communication}, Sci. China Inf. Sci., 67 (2024), 172205. 
	
	
	\bibitem{LiT2021decentralized}
	{\sc J.~Wang, T.~Li, and X.~Zhang}, {\em Decentralized cooperative online estimation with random observation matrices, communication graphs and time	delays}, IEEE Trans. Inform. Theory, 67 (2021), pp.~4035--4059.
	
	\bibitem{wang2023differentially}
	{\sc J. M. Wang, J. W. Tan, and J. F. Zhang}, {\em Differentially private distributed	parameter estimation}, J. Syst. Sci. Complex., 36 (2023),
	pp.~187--204.
	
	\bibitem{WangLY2003System} {\sc L. Y. Wang, J. F. Zhang, and G. Yin}, {\em  System identification using binary sensors}, IEEE Trans. Automat. Control, 48 (2003), pp.~ 1892-1907. 
	
	\bibitem{wang2019consensus}
	{\sc T.~Wang, H.~Zhang, and Y.~L. Zhao}, {\em Consensus of multi-agent systems	under binary-valued measurements and recursive projection algorithm}, IEEE Trans. Automat. Control, 65 (2019), pp.~2678--2685.
	
	\bibitem{wang2023event}
	{\sc Z.~Wang, C.~Jin, W.~He, M.~Xiao, G. P. Jiang, and J.~D. Cao}, {\em Event-triggered impulsive synchronization of heterogeneous neural networks}, Sci. China Inf. Sci., 5 (2024), 129202.
	
	
	\bibitem{wu2022dynamic}
	{\sc X.~Wu, B.~Mao, X.~Wu, and J.~L{\"u}}, {\em Dynamic event-triggered leader-follower consensus control for multiagent systems}, SIAM J. Control Optim., 60 (2022), pp.~189--209.
	
	\bibitem{xie2013LMS}
	{\sc S. L. Xie and H. R. Li}, {\em Distributed LMS estimation over networks with quantised communications},  Internat. J. Control, 86 (2013), pp.~478--492. 
	
	\bibitem{xie2018analysis}
	{\sc S.~Xie and L.~Guo}, {\em Analysis of normalized least mean squares-based consensus adaptive filters under a general information condition}, SIAM J. Control Optim., 56 (2018), pp.~3404--3431.
	
	\bibitem{ZhangQ2012}
	{\sc Q. Zhang and J. F. Zhang}, {\em Distributed parameter estimation over unreliable networks with Markovian switching topologies}, IEEE Trans. Automat. Control, 57 (2012), pp.~2545--2560. 
	
	
	\bibitem{zhao2018consensus}
	{\sc Y. L.~Zhao, T.~Wang, and W.~Bi}, {\em Consensus protocol for multiagent systems with undirected topologies and binary-valued communications}, IEEE Trans. Automat. Control, 64 (2019), pp.~206--221.
	
	\bibitem{zhu2018mitigating}
	{\sc S.~Zhu, C.~Chen, J.~Xu, X.~Guan, L.~Xie, and K.~H. Johansson}, {\em
		Mitigating quantization effects on distributed sensor fusion: A least squares
		approach}, IEEE Trans. Signal Process., 66 (2018), pp.~3459--3474.
	
	\bibitem{zhu2015distributed}
	{\sc S.~Zhu, Y.~C. Soh, and L.~Xie}, {\em Distributed parameter estimation with quantized communication via running average}, IEEE Trans. Signal Process., 63 (2015), pp.~4634--4646.
	
	\bibitem{zorich}
	{\sc V.~A. Zorich}, {\em Mathematical Analysis I}, 2nd ed.,  Springer,  New York, 2015. 
	
\end{thebibliography}
\end{document}